\documentclass{article}
\usepackage[margin=1in]{geometry}
\usepackage[hidelinks]{hyperref}
\usepackage[utf8]{inputenc}
\usepackage[small]{caption}
\usepackage{graphicx}
\usepackage{amsmath}
\usepackage{amsthm}
\usepackage{booktabs}
\usepackage{algorithm}
\usepackage{algorithmic}
\newif\ifhideproofs
\usepackage{apacite}
\usepackage{mathtools}

\usepackage{thmtools,thm-restate}

\usepackage{tabularx}
\usepackage{natbib}
\usepackage[shortlabels]{enumitem}
\usepackage[usenames,svgnames,xcdraw,table]{xcolor}

 \usepackage[capitalise,noabbrev]{cleveref}
\usepackage{float}
\usepackage{verbatim}
\usepackage{caption}

\usepackage{subcaption}
\newtheorem{lemma}{Lemma}
\newtheorem{theorem}{Theorem}
\newtheorem{corollary}{Corollary}
\newtheorem{definition}{Definition}

\newtheorem{example}{Example}

\newtheorem{proposition}{Proposition}

\usepackage{tikz}
\usetikzlibrary{decorations}
\usetikzlibrary{decorations.pathreplacing}

\usepackage{subcaption}

\title{Two-Sided Manipulation Games in Stable Matching Markets}

\author{
    \textbf{Hadi Hosseini} \\ 
    Penn State University\\ 
    \texttt{hadi@psu.edu}
    \and
    \textbf{Grzegorz Lisowski}\\
    AGH University of Science\\ and Technology \\ 
    \texttt{glisowski@agh.edu.pl} 
    \and 
    \textbf{Shraddha Pathak}\\
    Penn State University\\ \texttt{ssp5547@psu.edu}
}

\begin{document}

\newcommand{\shraddha}[1]{\textcolor{blue}{\textbf{Shraddha:} #1}}
\newcommand{\shraddhanew}[1]{\textcolor{blue}{#1}}

\newcommand{\grzesiek}[1]{\textcolor{red}{\textbf{Grzesiek:} #1}}
\newcommand{\grzegorznew}[1]{\textcolor{red}{#1}}

\newcommand{\DA}{\textsc{DA}}

\newcommand{\NE}{\textsc{NE}}
\newcommand{\NSP}{\textsc{NSP}}
\newcommand{\PoA}{\textsc{PoA}}
\newcommand{\PoS}{\textsc{PoS}}
\newcommand{\stratprof}{\textbf{p}}
\newcommand{\stratprrof}{\textbf{p'}}
\renewcommand{\>}{\succ}

\date{}
\maketitle

\begin{abstract}

The Deferred Acceptance ($\DA$) algorithm is an elegant procedure for finding a stable matching in two-sided matching markets. It ensures that no pair of agents prefers each other to their matched partners. In this work, we initiate the study of two-sided manipulations in matching markets as non-cooperative games. We introduce the \emph{accomplice manipulation game}, where a man misreports to help a specific woman obtain a better partner,  whenever possible. We provide a polynomial time algorithm for finding a pure strategy Nash equilibrium ($\NE$) and show that our algorithm always yields a stable matching---although not every Nash equilibrium corresponds to a stable matching. Additionally, we show how our analytical techniques for the accomplice manipulation game can be applied to other manipulation games in matching markets, such as one-for-many and the standard self-manipulation games. We complement our theoretical findings with empirical evaluations of different properties of the resulting $\NE$, such as the welfare of the agents.

\end{abstract}

\section{Introduction}
Many real-world markets, from school choice \citep{APR05new,APR+05boston} and resident matching \citep{roth1999redesign} to ride-sharing platforms \citep{banerjee2019ride} and recommendation systems \citep{eskandanian2020using}, are inherently markets with two sides. In these markets, there are two disjoint sets of agents, each with preferences over agents on the other side; for example, drivers and passengers in ride-sharing platforms, or prospective students and schools are in two separate sides of the market.\footnote{For the ease of exposition and consistency with the plethora of work in matching markets, in this paper, we refer to the two disjoint sets of agents as men and women.}
In such a scenario, the primary objective is to find a \emph{stable} matching, i.e., a matching in which no pair of agents prefer one another to their matched partners.

The Deferred Acceptance ($\DA$) algorithm \citep{gale1962college} is a celebrated mechanism that runs in rounds of proposals and rejections to output a stable matching. While $\DA$ is strategyproof for the proposing agents (i.e., men) \citep{dubins1981machiavelli,huang2006cheating}, it is known to be susceptible to strategic misreporting by the proposed-to side (i.e., women) \citep{roth1999truncation,teo2001gale}. In fact, not only $\DA$, but all stable matching mechanisms are prone to  strategic manipulation \citep{roth1982economics}. 

The simplicity of the $\DA$ algorithm has prompted studies on the different possible types of manipulations, with a major focus on \emph{one-sided} manipulations, where coalitions of misreporting agents are from the same side of the market \citep{dubins1981machiavelli,huang2006cheating}. Such manipulations are observed in the real world. For instance, drivers on ride-hailing platforms collude to create an artificial shortage, triggering surge pricing. 
In addition, often the coalition of manipulators contains agents from \textit{both} sides of the market: drivers sometimes collude with passengers by opting for `offline rides' upon arrival. Agents on one side of the market may also indirectly influence the behavior of those on the other side; for instance, drivers influence the behavior of the riders by moving further away \citep{banerjee2019ride}, or schools influence the behavior of `undesirable' students by appearing less attractive by hiking certain fees or mandating uniforms \citep{HKN16improving}.


Such examples have motivated the study of coalitions consisting of agents from both sides of the market, i.e., \emph{two-sided} manipulations, in the recent past \citep{BH19partners,hosseini2021accomplice,hosseini2022twoforone}. 
Specifically, \citeauthor{BH19partners}~[\citeyear{BH19partners}] and \citeauthor{hosseini2021accomplice}~[\citeyear{hosseini2021accomplice}]
consider \emph{manipulation via an accomplice} in the $\DA$ algorithm where an accomplice man may misreport his preferences to help a specific woman, without harming himself. Or, a man may misreport his preferences to help \emph{all} women on the other side \citep{hosseini2022twoforone}.


These studies focus on finding the \emph{optimal misreport} and understanding their \emph{effect on stability} of the resulting matching. However, they do not consider the potential subsequent strategic incentives that they create or the effect of uncoordinated simultaneous misreports by multiple agents. 
Such strategic behavior by multiple coalitions of agents gives rise to two-sided manipulations as non-cooperative games in matching markets, raising the questions of whether these games have an equilibrium, and if they do, what their properties are.
In the ride-hailing example, although offline rides may initially benefit individual driver-passenger pairs, this advantage encourages more pairs to engage in manipulation, ultimately reducing the ride-hailing company's profit margins and resulting in higher prices for passengers. 


\paragraph{Our Results and Techniques.} 
We initiate the study of two-sided manipulation games by investigating two variants of the problem: \textit{accomplice} manipulation games and \textit{one-for-many} manipulation games.
In these games, strategic men misreport their preferences (through \textit{permutation}) to help a strategic woman (or women) receive a better match, without harming themselves. Such pairs of men and women are called `strategic pairs', and a reported preference profile by all strategic men is said to be a (pure strategy) Nash equilibrium ($\NE$) when no such strategic pair can manipulate. 

Our first main result (\cref{cor:NE_existence}) shows that the accomplice manipulation game always supports a pure strategy $\NE$.
Our constructive proof, based on careful construction of the best-response dynamics in the associated game,  enables us to devise an algorithm to find an $\NE$ in polynomial time.
Although not every best-response dynamic converges (\cref{ex:unstable}), we show that dynamics consisting of optimal (or even sub-optimal) \textit{push-up} strategies always do. These strategies ensure that the resulting matching in every step of the dynamic lies within the original stable lattice, thereby providing a potential function that guarantees (polynomial time) convergence.
This also guarantees that the $\NE$ preference profile found by our algorithm results in a stable matching (with respect to the truthful profile) under $\DA$, even though not all $\NE$s correspond to a stable matching. Interestingly, all of them are approximately stable (\cref{thm:somestable}).

In \cref{sec:transferable}, we show how our technique of adroit construction of best-response dynamics can be utilized to find an $\NE$ in one-for-all manipulation \citep{hosseini2022twoforone} and self-manipulation by a woman \citep{teo2001gale}.
While the existence of $\NE$ was shown previously in woman manipulation games \citep{roth1984misrepresentation,zhou1991stable,gupta2016total}, our technique reveals insights into the structure of Nash equilibria and provides a direct construction of polynomial-time algorithms for computing an $\NE$.

In \cref{sec:experiments}, we complement our theoretical results with an experimental evaluation of more properties of the $\NE$; namely welfare of agents in $\NE$ and length of the best-response dynamic. 

All omitted proofs along with additional results can be found in the supplementary material.

\paragraph{Related Literature.} The incompatibility of stability and strategy-proofness in two-sided matching markets \citep{roth1982economics} has led to extensive studies on manipulations of the $\DA$ algorithm. The distinctions are based on two factors: the type of misreport---truncating preferences \citep{RR99truncation} or permuting preferences \citep{teo2001gale}---and the misreporting agents, whether it is a single agent \citep{dubins1981machiavelli,teo2001gale} or a coalition \citep{huang2006cheating,shen2018coalition,shen2021coalitional}. 

Traditionally, the focus on coalitional manipulations has been on coalitions with agents from the same side of the market. 
\citet{BH19partners} and \citet{hosseini2021accomplice} initiated the study of coalitions consisting of agents from different sides of the market, and showed that optimal manipulation strategies can be computed in polynomial time for accomplice \cite{hosseini2021accomplice} and one-for-all manipulations \cite{hosseini2022twoforone}.

The study of misreports in $\DA$ as non-cooperative games has also predominantly focused on one-sided manipulations. Specifically, \citet{roth1984misrepresentation} showed that for the single-agent manipulation game by women, every $\NE$ results in a stable matching. Furthermore, every $\NE$ can be characterized in the following manner: Whenever a stable matching is supported by some preference profile, there also exists an $\NE$ that attains this matching \citep{zhou1991stable}. Since $\DA$ on the truthful preference profile results in a stable matching, i.e., there exists a stable matching and a preference profile (truthful) that supports it, the \emph{existence} of a Nash equilibrium for this game is known. Its computation was later studied by \citet{GIM16total}. 
Even for games in which the players are a coalition of women, a strong $\NE$ is known to exist, and the resulting matching is unique \citep{shen2018coalition}.

\section{Preliminaries} \label{sec:notation}

\paragraph{Stable Matching.} A \emph{stable matching} instance $\langle M, W, \succ \rangle$ consists of two disjoint sets of agents, colloquially referred to as  \emph{men} ($M$) and \emph{women} ($W$) where $|M|=|W|=n$, as well as a \emph{preference profile} $\succ$ that specifies the \emph{preference lists} of all agents. The preference list of an agent $i$, denoted by $\succ_i$, is a strict total ordering over agents on the other side. So, we write $m_i \succ_w m_j$ if a woman $w$ prefers man $m_i$ to $m_j$ and $m_i \succeq_w m_j$ if she is indifferent between them or (strictly) prefers $m_i$ to $m_j$. %
Similarly, for each man $m\in M$, $w_i \succ_m w_j$ and $w_i \succeq_m w_j$ indicate a strict and weak preference of $m$ for $w_i$ over $w_j$.
Moreover, we write $\succ_{-X}$ to denote the preference profile of all agents apart from the agents in the set $X$, and thus, $\succ = (\succ_{-X}, \succ_{X})$.

A \emph{matching} is a bijective function $\mu:\, M \cup\ W \to M \cup\ W$, where, for every $m\in M$ and $w\in W$, $\mu(m)\in W$, $\mu(w)\in M$, and $\mu(w)=m$ if and only if $\mu(m)=w$. 
For a matching $\mu$, a pair of agents $(m,w) \in M\times W$ forms a \emph{blocking pair}, if $w \succ_m \mu(m)$ and $m \succ_w \mu(w)$, i.e., when $m$ and $w$ prefer each other to their matched partners in $\mu$. 
A matching is said to be \emph{stable} if it does not contain a blocking pair.
The set of all stable matchings is denoted by  $\mathcal{S}_{\succ}$ and forms a distributive lattice that is possibly exponential in size \citep{K97stable}.

Given a subset of agents $X\subseteq M\cup W$, for any two matchings $\mu,\mu'$ we write $\mu \succeq_X \mu'$ if all its members weakly prefer $\mu$ over $\mu'$, i.e., $\mu \succeq_i \mu'$ for all $i \in X$. We write $\mu \succ_X \mu'$ when $\mu \succeq_X \mu'$ and at least one agent $j\in X$ strictly prefers $\mu$ to $\mu'$, i.e., $\mu \succ_j \mu'$.

\paragraph{The Deferred Acceptance  Algorithm.
}
The \emph{deferred acceptance algorithm} ($\DA$) guarantees to find a stable matching $\mu$ in a two-phase procedure. First, in the \emph{proposal phase}, each currently unmatched man proposes to his favorite woman who has not rejected him yet. 
Subsequently, in the \emph{rejection phase}, each woman tentatively accepts her favorite proposal, rejecting the others. The procedure terminates when no more proposals are possible. We denote the outcome of $\DA$ under the preference profile $\succ$ by $\mu_\succ \coloneqq \DA(\succ)$. 
\citeauthor{gale1962college}~[\citeyear{gale1962college}] showed that given any profile $\succ$ the $\DA$ algorithm always returns a matching that is stable and \emph{men-optimal}, i.e., for each $\mu' \in S_\succ$,  $\mu_\succ \succeq_M \mu'$.  
In contrast, $\DA$ is \emph{women-pessimal}, i.e., for every $\mu' \in S_\succ$,  $\mu' \succeq_W \mu_\succ$ \citep{mcvitie1971stable}.

\paragraph{One-Sided Strategies.}
For a profile $\>$ and the matching $\mu_\succ\coloneqq \DA(\succ)$, a woman $w$ can \textit{self manipulate} if there is a misreport $\>'_w$ (a \textit{permutation} of $w$'s true list $\>_{w}$) such that $\mu_{\succ'} \>_w \mu_\succ$ where $\mu_{\succ'} \coloneqq \DA{(\>_{-w}, \>_{w'})}$. 
The self manipulation strategies are \textit{one-sided}, i.e., a manipulator woman misreports to her own benefit. These strategies are known to admit tractable algorithms~\citep{teo2001gale,vaish2017manipulating}. 

\paragraph{Two-Sided Strategies.}
A manipulation strategy could be \textit{two-sided} involving agents from both sides of the market where an agent from one side (a man) misreports to improve the match of beneficiary agents (a set of women).
Formally, given a profile $\>$, we say that a woman $w$ can manipulate through an \textit{accomplice} $m$ if $\mu_{\succ'} \>_{w} \mu_\succ$ where $\mu_\succ \coloneqq \DA{(\>)}$ and $\mu_{\succ'} \coloneqq \DA{(\>_{-m}, \>'_{m})}$.
A generalization of this two-sided strategy is \textit{one-for-many} manipulation where a misreporting agent improves the match of a subset of women. Formally, a strategy is \textit{one-for-many} if, 
for a set of beneficiary women $P_w \subseteq W$, man $m$'s misreport $\succ'_m$ is such that $\mu_{\succ'} \succ_{P_w} \mu_\succ$, where $\mu_{\succ'} \coloneqq \DA{(\>_{-m}, \>'_{m})}$ and $\mu_\succ \coloneqq \DA(\succ)$.
These two-sided manipulation strategies were introduced by \citeauthor{hosseini2021accomplice}~[\citeyear{hosseini2021accomplice}] and \citeauthor{hosseini2022twoforone}~[\citeyear{hosseini2022twoforone}]. 

We say that an accomplice $m$ incurs \textit{regret} if he receives a less preferred match after misreporting, i.e., $\mu_\succ \>_{m} \mu_{\succ'}$. Due to the strategyproofness of $\DA$ for the men \citep{dubins1981machiavelli}, given any misreport $\>'_{m}$, we have $\mu_\succ \succeq_m \mu_{\succ'}$. Thus, a manipulation strategy is \textit{with-regret} if $\mu_{\succ'}(m) \neq \mu_\succ(m)$. Otherwise, man $m$ receives its original outcome and incurs \textit{no-regret}, i.e. $\mu_{\succ'}(m) = \mu_\succ(m)$.

\paragraph{Stability Relaxations.}
Given a preference profile $\>$ and a fixed set of pairs $X\subseteq M\times W$, we say that a matching $\mu$ is $X$-stable with respect to $\>$ if every blocking pair (if one exists) belongs to $X$. 
Clearly a stable matching is $X$-stable, and every matching is $X$-stable when $X = M\times W$.
We note that this relaxation generalizes the notion of $m$-stability \citep{BH19partners}; an $m$-stable matching is $X$-stable with $X=\{m\}\times W$. Under accomplice manipulation, any matching that is stable with respect to the manipulated profile $\>'$ is $X$-stable with respect to the true profile $\>$ and the set $X$ is exactly $\{m\}\times W$ \citep{hosseini2021accomplice}.

\section{Two-Sided Manipulation Games}

Here, we formally define two variants of two-sided manipulation games, namely \emph{accomplice manipulation game} and \emph{one-for-many manipulation game}. 
Also, we introduce additional operations and theoretical techniques that will be used to analyze these games.

\paragraph{Accomplice Manipulation Game.}
Given an instance $\langle M, W, \succ\rangle$, the players of the game are given by a set of 
\emph{strategic pairs}, denoted by $P \subseteq M \times W$. Let $P_m$ be the set of \emph{strategic men}, i.e., $ P_m \coloneqq \{m\in M \mid (m,w)\in P \text{ for some } w\in W \}$, and $P_w \coloneqq \{w\in W \mid (m,w)\in P \text{ for some } m\in M \}$ represent the set of \emph{strategic women}. 
Note that pairs in $P$ need not be pairwise disjoint. That is, an agent can be in alliance with multiple individuals but not as a coalition, i.e., there is no coordination among accomplices or beneficiaries.

 A \emph{strategy profile} $\stratprof = (\stratprof_1, \ldots, \stratprof_{n})$ is a list of \textit{reported} (not necessarily true) total orderings by the accomplices, i.e., the men.
 We assume that every woman reports her preference list truthfully ($\succ_w)$, and thus, for ease of exposition, we omit their preference list and write $\DA(\stratprof)$ to denote the outcome of \DA{} on the preference list of all agents including those who report truthfully.

For a strategy profile $\stratprof$, an \emph{accomplice manipulation} by a strategic pair $(m,w)\in P$ is a misreported preference $\stratprrof_m$ by a manipulator $m$, such that, for $\mu_\stratprof\coloneqq \DA(\stratprof)$ and $\mu_{\stratprrof} \coloneqq \DA(\stratprof_{-m}, \stratprrof_m)$, 
(i) the beneficiary woman $w$ benefits from this misreport, i.e., $\mu_{\stratprrof} \succ_{w} \mu_\stratprof$, and
	(ii) the manipulator man $m$ does not become worse-off (with respect to his true preferences) from this misreport, i.e., $\mu_{\stratprrof} \succeq_{m} \mu_\stratprof$. Note that in conditions (i) and (ii), the partners in matchings corresponding to the profiles $\stratprof$ and $\stratprrof$ are compared according to the true profile $\succ$. Such a manipulation is called a \emph{better response} to the profile $\stratprof$, and it is the \emph{best response} if it provides the best partner for $w$ among all better responses.

For an accomplice manipulation game with strategic pairs $P\subseteq M\times W$, a strategy profile $\succ_\NE$ is  a \emph{Nash Equilibrium} ($\NE$) if there does not exist a strategic pair $(m,w)\in P$, such that $(m,w)$ can perform an accomplice manipulation at the profile $\succ_\NE$.

\paragraph{One-for-Many Manipulation Game.}
Given a strategy profile $\stratprof$, a \emph{one-for-many manipulation} by a man $m$ for a subset of beneficiary women $P^m_w \subseteq W$ is a misreported preference $\stratprrof_m$  such that, for $\mu_\stratprof \coloneqq \DA(\stratprof)$ and $\mu_{\stratprrof} \coloneqq \DA(\stratprof_{-m}, \stratprrof_m)$, we have that 
(i) $\mu_{\stratprrof}$ Pareto dominates $\mu_{\stratprof}$ for the set of women $P^m_w$, i.e.,  for every $w\in P^m_w$, $\mu_{\stratprrof} \succeq_{w} \mu_{\stratprof}$ and for some $w'\in P^m_w$, $\mu_{\stratprrof} \>_{w'} \mu_{\stratprof}$, and 
 (ii) the manipulator man $m$ is not worse-off through this misreport, i.e., $\mu_{\stratprrof} \succeq_{m} \mu_{\stratprof}$. Such a manipulation $\stratprrof$ is said to be a better response to $\stratprof$, and it is the best response if no other strategy $\stratprof_m''$ Pareto dominates $\stratprrof_m$.

Then, in the context of a one-for-all manipulation game with $P_m\subseteq M$ as a set of misreporting men and $\{P_w^m\}_{m\in P_m}$ as the set of beneficiary women, a strategy profile $\succ_\NE$ is said to be a {Nash equilibrium} ($\NE$)  if there is no misreporting agent $m\in P_m$ such that $m$ can perform a one-for-many strategy at $\succ_\NE$.

\paragraph{Push-up Operation.}

Let $\stratprof$ be a strategy profile (not necessarily $\succ$) and $\mu_\stratprof \coloneqq \DA(\stratprof)$. 
For any man $m$, let $\mathcal{L}(\stratprof_{m},\mu_\stratprof)$ and $\mathcal{R}(\stratprof_{m},\mu_\stratprof)$ denote the set of women $m$ reports above and below his $\DA(\stratprof)$ partner in $\stratprof_m$, i.e., 
$\stratprof_{m} = (\mathcal{L}(\stratprof_{m},\mu_\stratprof), \mu_{\stratprof}(m), \mathcal{R}(\stratprof_{m},\mu_\stratprof))$.
We say that a man $m$ performs a \emph{push-up} operation on a set $X\subseteq W$ if the new list is $\stratprof_{m}^{X\uparrow} = (\mathcal{L}(\stratprof_{m},\mu_\stratprof)\cup X, \mu_\stratprof(m), \mathcal{R}(\stratprof_{m},\mu_\stratprof)\setminus X)$; see \cref{fig:push-up} for a pictorial representation.
 \begin{figure}[t]
    \centering
    \scalebox{.8}{ \begin{tikzpicture}
         \draw (0, 0) -- (7, 0);
        
         \node at (1.75, 0.6) {$\mathcal{L}(\stratprof_m, \mu_\stratprof)$};
        
         \node at (5.25, 0.6) {$\mathcal{R}(\stratprof_m, \mu_\stratprof)$};
        
         \node at (3.5, -0.5) {$\mu_\stratprof(m)$};
        
         \node at (-0.5, 0) {$\stratprof_m$};

         \draw (3.5, -0.15) -- (3.5, 0.15);
        
         \draw (4.75, -0.2) rectangle (5.75, 0.2);
         \node at (5.25, -0.5) {$X$};

         \draw[->, >=stealth, line width=0.5pt] (5.25, -0.8) -- (5.25, -1) -- (2, -1) -- (2, -0.6);
     \end{tikzpicture}}
     \caption{Push-up operations by man $m$ in strategy profile $\stratprof$.}
     \label{fig:push-up}
 \end{figure}
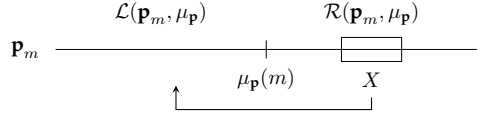
Note that $\stratprof_m^{X\uparrow}$ does not encode a unique preference order, but every permutation of preferences above or below the \DA{} match provides the same matching \citep{huang2006cheating}. For a more detailed discussion, we refer the reader to \cref{sec:appendix_figures}.

Notably, any accomplice manipulation can be performed through push-up manipulations only 
(\cref{prop:accomplice_properties}(i)). Also, as long as the manipulator $m$ incurs no-regret, the stable lattice of the manipulated instance consisting of push-up operations is a subset of the original one (\cref{prop:accomplice_properties}(ii)). 
Thus, each woman gets weakly better-off and each man weakly worse-off from a no-regret accomplice manipulation (\cref{prop:accomplice_properties}(iii)). 

\begin{proposition}[\cite{hosseini2021accomplice}]\label{prop:accomplice_properties}
    Let $\stratprof$ be a strategy profile. For any accomplice manipulation $\stratprrof$ by accomplice man $m$ such that $\mu_{\stratprrof} \coloneqq \DA{(\stratprrof)}$, we have the following:  
    \begin{enumerate}[(i)]
        \item $\stratprrof$ can be performed through an inconspicuous push-up operation on $\stratprof$, i.e., $\stratprrof := (\stratprof_{-m}, \stratprof_m^{X\uparrow})$ where $|X|=1$.
        \item For a $\stratprrof := (\stratprof_{-m}, \stratprof_m^{X\uparrow})$, we have $\mathcal{S}_{\stratprrof} \subseteq \mathcal{S}_{\stratprof}$ as long as $m$ incurs no regret with respect to $\stratprof$, i.e., $\mu_{\stratprrof} \succeq^{\stratprof}_m \mu_{\stratprof}$.
        \item For a $\stratprrof := (\stratprof_{-m}, \stratprof_m^{X\uparrow})$ and its corresponding matching $\mu_{\stratprrof}:=\DA(\stratprrof)$, if $m$ does not incur regret (i.e., $\mu_{\stratprrof} \succeq^{\stratprof}_m \mu_{\stratprof}$), then $\mu_{\stratprrof} \succ^{\stratprof}_W \mu_\stratprof$ and $\mu_\stratprof \succ^\stratprof_M \mu_{\stratprrof}$.
    \end{enumerate}
\end{proposition}
Here, $\mu_\stratprrof \succ^\stratprof_i \mu_{\stratprof}$ represents that, in the strategy profile $\stratprof$, agent $i$ prefers its partner in $\mu_\stratprrof$ to its partner in $\mu_\stratprof$. Similarly interpret $\succ_X^\stratprof$ for the set of agents $X$.

\section{Accomplice Manipulation Games} \label{sec:accomplice}

Here, we consider accomplice manipulation games that may involve \textit{any} subset of strategic pairs. 
The primary questions are 
(i) the existence of an NE and
(ii) computing an $\NE$ (if one exists). 
Once we address the above questions, we will discuss the properties of an $\NE$ including its  stability with respect to true preferences and the number of potential equilibria.

First, we highlight a relationship between the $\NE$ of two accomplice games. 
Take two subsets of strategic pairs $P, P' \subseteq M\times W$ such that $P_m=P_m'$ and $P_w \subseteq P_w'$.
\cref{lem:NEsubsets} implies that the accomplice game with $P$ as strategic pairs have at least as many $\NE$ as the game with $P'$ as strategic pairs. 
Also, finding a $\NE$ for the game with $P_m\times W$ as the set of strategic agents suffices to find an $\NE$ for any game with $P=(P_m, P_w)$. 

\begin{restatable}{proposition}{lemNEsubsets}
\label{lem:NEsubsets}
Consider two sets of strategic pairs $P=(P_m, P_w)$ and $P'=(P'_m, P_w')$ such that $P_w \subseteq P_w'$ and $P_m=P_m'$. If a strategy profile $\stratprof$ is an $\NE$ for the game with $P'$, then it is also an $\NE$ for the game with $P$.
\end{restatable}
A common method to prove that a pure strategy $\NE$ always exists for any scenario is to
demonstrate that better or best responses do not cycle. This ensures that best (or better) response dynamics converges to an NE. 
However, for our game, it turns out that such dynamics may not necessarily converge when accomplice agents' strategies consist of \emph{any} misreport; see \cref{ex:unstable} (adapted from \citep{hosseini2021accomplice}). 

\begin{example}[Best-response dynamics may not converge] \label{ex:unstable}
  
    Consider an instance with five men and five women, with preferences as listed below.  Let $P=\{(m_3, w_4), (m_3, w_1)\}$ be the strategic pairs.

    \begin{table}[ht]
\centering 
    \begin{tabularx}{0.9\linewidth}{XXXXXXXXXXXXXXXXXXXXXXX}
    
         	$m_1$: & $w_2^*$& \underline{$w_1$} & $w_3$ & $w_4$ & $w_5$ && $w_1$: & \underline{$m_1$} & $m_3$ & $m_2^*$ & $m_4$ & $m_5$\\
            $m_2$: & {$w_1^{*}$} & \underline{$w_2$} & $w_3$ & $w_4$ & $w_5$ && $w_2$: & \underline{$m_2$} & {$m_1^{*}$} & $m_3$ & $m_4$ & $m_5$\\
            $m_3$: &{$w_1$} &  \underline{$w_3^*$} & $w_4$ & $w_2$ & $w_5$ && $w_3$: & \underline{$m_3^*$} & $m_1$ & $m_2$ & $m_4$ & {$m_5$}\\
            $m_4$: & \underline{$w_4$} & $w_5^*$ & $w_1$ & $w_2$ & $w_3$ && $w_4$: & {$m_5^{*}$} & $m_3$ & $m_1$ & $m_2$ & \underline{$m_4$}\\
            $m_5$: & \underline{$w_5$} & $w_4^*$ & $w_1$ & $w_2$ & $w_3$ && $w_5$: & $m_4^*$ & $m_1$ & $m_2$ & {$m_3$} & \underline{$m_5$}
    \end{tabularx}
\vspace{-0.06in}
\end{table}

    The $\DA$ matching on $\succ$ is underlined. Notice that the strategic profile $\succ'$ where $m_3$ reports $w_4\succ' w_3\succ'w_1 \succ' w_2 \succ' w_5$ and everyone else reports truthfully results in the matching $*$ and is an accomplice manipulation for $(m_3, w_4)$. However, at $\succ'$, $(m_3, w_1)$ is an accomplice pair, and $\succ$ is an accomplice manipulation for $(m_3, w_4)$. 
\end{example}

In the above example, notice that $\succ'$ resulted from an operation which was \emph{not} a push-up operation on $\succ_{m_3}$; woman $w_1$ is pushed-down in the preference list $\succ'_{m_3}$. In fact this is precisely the reason for the cyclic nature of this dynamic. 
Next, we show that best and better 
response dynamics consistent with push-up manipulation strategies, defined for any set of strategic pairs $P\subseteq M\times W$, converge in polynomial time, thereby proving the existence of a pure strategy $\NE$.

\subsection{Existence and Computation of an NE}\label{sec:NE}

Our constructed best (better) response dynamics starts at the truthful preference profile and each successor profile is a result of an optimal (or sub-optimal) push-up accomplice manipulation by some strategic pair $(m,w)\in P$.

\begin{definition}[Push-up Dynamics]\label{def:best_res}
For strategic pairs $P$, a \emph{push-up dynamic} is a sequence of strategy profiles $\overline{\stratprof} =(\stratprof^{(1)}, \ldots)$ starting at $\stratprof^{(1)} =\, \succ$, where each subsequent profile $\stratprof^{(t+1)}$ results from a no-regret,  \underline{push-up} accomplice manipulation by some strategic pair $(m,w)\in P$ at $\stratprof^{(t)}$, if possible. 
Otherwise, $\stratprof^{(t+1)} = \stratprof^{(t)}$ and we say that the push-up dynamic \emph{converges} with $\stratprof^{(t)}$ as its \emph{fixed-point}.
\end{definition}

Using this best-response dynamic, we show that we can efficiently find an $\NE$, as stated below. 

\begin{theorem}[Existence of an $\NE$] \label{cor:NE_existence}
For every accomplice manipulation game defined by a set of strategic players $P\subseteq M\times W$, there exists a polynomially computable $\NE$. 
\end{theorem}

Notice that in the push-up dynamic every subsequent manipulation is with respect to the current profiles. This ensures that, if the dynamic converges, then no push-up accomplice manipulations can be performed on the fixed point (manipulations via other operations may still be possible). To show the existence and efficient computation of an $\NE$ for the accomplice game we need to show the following two properties of the push-up dynamics: (i) it always converges in polynomial time, and
    (ii) its fixed point does not allow for \emph{any} (not just push-up) accomplice manipulation by the strategic pairs.  

For showing the convergence of the push-up dynamics, we prove that as we move along the dynamic sequence, (i) the set of stable solutions shrinks, and (ii) for any man $m$, the set of women preferred lower than $m$'s partner in matching $\mu_{\stratprof^{(t+1)}}$ in $\stratprof_m^{(t+1)}$ is a subset of the set of women preferred
lower than $m$’s partner in the same matching but in the preference profile $\stratprof_m^{(t)}$.

\begin{restatable}{lemma}{lemshrinkingset}
\label{lem:shrinking set}
    In the push-up dynamics, the following set inclusions hold for each man $m$:
    \begin{enumerate}[(i)]
        \item $\mathcal{S}_{\stratprof^{(t+1)}} \subseteq \mathcal{S}_{\stratprof^{(t)}}\subseteq \ldots \subseteq \mathcal{S}_{\stratprof^{(1)}} = \mathcal{S}_{\succ}$, and
        \item $\mathcal{R}(\stratprof^{(t+1)}_m, \mu_{\stratprof^{(t+1)}} )\subseteq \mathcal{R}(\stratprof^{(t)}_m, \mu_{\stratprof^{(t+1)}}) \subseteq \ldots \subseteq \mathcal{R}(\succ_m, \mu_{\stratprof^{(t+1)}})$.
    \end{enumerate} 
\end{restatable}

Notice that part (ii) of Lemma \ref{lem:shrinking set} implies that if we perform a manipulation at $\stratprof^{(t+1)}$ and the manipulation is such that $\mu_{\stratprof^{(t)}} \succ^{\stratprof^{(t)}}_m \mu_{\stratprof^{(t+1)}}$ for the manipulator $m$, i.e., $\mu_{\stratprof^{(t+1)}}(m) \in \mathcal{R}(\stratprof^{(t+1)}_m, \mu_{\stratprof^{(t)}})$, then $\mu_{\stratprof^{(t)}} \succ_m \mu_{\stratprof^{(t+1)}}$. That is, the manipulation was with-regret since $m$ became worse-off. So, in the push-up dynamics, it never happens that $\mu_{\stratprof^{(t)}} \succ^{\stratprof^{(t)}}_m \mu_{\stratprof^{(t+1)}}$.

\begin{restatable}{lemma}{actualiscurrent}
    \label{lem:actualiscurrent}
In the push-up dynamics, if $\stratprof^{(t)} \neq \stratprof^{(t+1)}$, then $\mu_{\stratprof^{(t+1)}} \succeq^{\stratprof^{(t)}}_m \mu_{\stratprof^{(t)}}$.
\end{restatable}

This implies that computing an accomplice manipulation at $\stratprof^{(t)}$ reduces to finding a push-up manipulation in the instance $\langle M, W, (\succ_{-S_m}, \stratprof^{(t)}) \rangle$.
With this, we can show that the push-up dynamic converges.

\begin{lemma}[Push-up dynamics always converge] \label{th:brseq converges}
    For every accomplice manipulation game defined by the set of strategic players $P \subseteq M\times W$, the push-up dynamics always converge in a polynomial number of steps. 
\end{lemma}

\begin{proof}
If $\stratprof^{(t+1)}\neq\stratprof^{(t)}$, then there is some accomplice pair $(m,w)\in P$ who can perform a push-up manipulation at $\stratprof^{(t)}$. By Lemma \ref{lem:actualiscurrent}, this means that the pair $(m, w)$ can no-regret accomplice manipulate in the instance $\langle M, W, \stratprof^{(t)}\rangle$. Then, by Proposition \ref{prop:accomplice_properties}, 
 $   \mu_{\stratprof^{(t+1)}} \succ_W^{\stratprof^{(t)}} \mu_{\stratprof^{(t)}}$.
As in the push-up dynamics of accomplice manipulations women do not change their preference profiles,  $\succ_W^{\stratprof^{(t)}} =\, \succ_W$ for every $i$ in the sequence.
Hence, by a potential argument, the set of women can be better-off with respect to $\succ$ at most $n$ times. Thus, at some point $\stratprof^{(t+1)}=\stratprof^{(t)}$.
Note that at every step $i$ in which $\stratprof_{i+1} \neq \stratprof^{(t)}$, at least one woman gets strictly better-off. As an individual woman can improve her matching at most $n$ times, 
the sequence converges by the step $i \leq n^2$. 
\end{proof}

As the push-up dynamics always converge, to prove that the fixed point $\stratprof^*$ is an $\NE$, it remains to show that no accomplice manipulation (not just push-up) is possible at $\stratprof^*$. 

\begin{proof}[Proof of Theorem \ref{cor:NE_existence}]
    By Lemma \ref{th:brseq converges}, the push-up dynamic converges, i.e., $\stratprof^{(t)}=\stratprof^{(t+1)}$. When $\stratprof^{(t)}=\stratprof^{(t+1)}$,  $\stratprof^{(t)}$ is a fixed point and there is no strategic pair $(m,w)\in P$ who can perform a push-up accomplice manipulation. So, if the converged point $\stratprof^{(t)}$ is not an $\NE$, then some pair $(m,w) \in P$ can perform an accomplice manipulation at $\stratprof^{(t)}$ via an operation which is not push-up. Let the manipulated strategy profile be $\stratprrof$. As it is a no-regret manipulation, $\mu_{\stratprrof} \succeq_m \mu_{\stratprof^{(t)}} $, i.e., $\mu_{\stratprrof} \notin \mathcal{R}(\succ_m, \mu_{\stratprof^{(t)}})$. From this and Lemma \ref{lem:shrinking set}, we have that $\mu_\stratprrof\notin \mathcal{R}(\stratprof^{(t)}_m, \mu_{\stratprof^{(t)}})$, i.e., $\mu_\stratprrof \succeq^{\stratprof^{(t)}}_m \mu_{\stratprof^{(t)}}$. So, if we consider the instance $\langle M, W, \stratprof^{(t)}
    \rangle$, then the manipulation $\stratprrof$ is an accomplice manipulation on this instance. But, by Proposition \ref{prop:accomplice_properties}, if $m$ can manipulate through an operation which is not push-up, then he can also manipulate through an inconspicuous push-up operation. But, as $\stratprof^{(t)}$ is a fixed point, no (no-regret) push-up manipulations are possible. So, $\stratprof^{(t)}$ is an NE.
    Also, as the sequence converges in at most $n^2$ steps and finding a successor in the sequence is polynomial-time solvable (as shown in \cref{sec:appendix_accomplice}), an $\NE$ can be found in polynomial-time.
\end{proof}


\subsection{Properties of $\NE$
} \label{sec:properties_NE}

Here, we discuss several properties of $\NE$ in the accomplice manipulation game.

\paragraph{Stability.} 
Given the existence of an $\NE$ in accomplice manipulation games, a natural question is to ask whether, for every $\NE$ profile $\succ_{\NE}$, the matching $\DA(\succ_\NE)$ is stable with respect to true preferences. 
We address this question by showing that (i) the $\NE$ constructed in \cref{sec:NE} through the push-up dynamics results in a stable matching, and (ii) while all equilibria do not necessarily correspond to a stable matching, every $\NE$ is \emph{$X$-stable}, for $X=(M\times W)\setminus P$.

\begin{restatable}{theorem}{accomplicestable}
    \emph{(Stability of Nash Equilibria)}
\label{thm:somestable}
For an accomplice manipulation game with $P\subseteq M\times W$:
\begin{enumerate}[(i)]
    \item The matching corresponding to the $\NE$ profile constructed from the push-up dynamic is stable.
    \item Some $\NE$ profiles  result in an unstable matching with respect to true preferences.
    However, every matching corresponding to an $\NE$ profile is $\big((M\times W )\setminus P\big)$-stable.  
\end{enumerate}
\end{restatable}

Thus, when $P=M\times W$, every $\NE$ results in a stable matching under $\DA$.

\begin{proof}
From Theorem \ref{cor:NE_existence}, the push-up dynamics converges to an $\NE$. By Lemma \ref{lem:shrinking set}, the set of stable matchings at each step of the dynamics satisfies:
        \begin{align*}
            \mathcal{S}_{\stratprof^{(t)}} \subseteq \mathcal{S}_{\stratprof^{(t-1)}} \ldots \subseteq \mathcal{S}_{\stratprof^{(1)}} = \mathcal{S}_{\succ}.
        \end{align*} 
        
        Let $\stratprof^{(t)}$ be the fixed point of the dynamic. The corresponding matching $\mu_{\stratprof_{i}} = \DA(\stratprof_{i}, \succ_{-P_m})$ belongs to the set $\mathcal{S}_{\stratprof^{(t)}}$. From the containment property of the set of stable solutions, $\mu_{\stratprof^{(t)}} \in \mathcal{S}_{\succ}$. The claim follows.
        
        For the second part, to show that some NE profiles result in an unstable matching, consider the instance in \cref{ex:unstable} with strategic pairs $P=\{(m_3,w_4)\}$. Notice that the profile $\succ' = (\succ_{-m_3}, \succ'_{m_3})$ is an $\NE$, but the resulting matching $\mu_{\succ'}=\DA(\succ')$ (marked by $^*$) is unstable since $(m_3, w_1)$ forms a blocking pair.
        
        Proof of $((M \times W) \setminus P)$-stability: Let $\succ_\NE$ be a $\NE$ profile for the accomplice game with the strategic pairs $P\subseteq M\times W$, and let its corresponding matching be $\mu_{\succ_\NE}= \DA(\succ_\NE)$.
        Assume, for contradiction, that $(m,w)\in P$ is a blocking pair in $\mu_{\succ_\NE}$. Consider the modified preference profile $\succ' = ({\succ_\NE}_{-m}, \succ'_{m})$ where $m$'s preferences are modified to place $w$ at the top, i.e. $\succ'_m = (w, \ldots)$. In $\mu_{\succ'}=\DA(\succ')$, $m$ is matched to $w$; that is, through the misreport $\succ'_m$, $m$ incurs no-regret and $w$ benefits from the misreport, i.e. $\succ'$ is an accomplice manipulation at $\succ_\NE$. This contradicts the assumption that $\mu_{\NE}$ is an $\NE$.
\end{proof}

The loss in stability in $\NE$ can also be measured through the price of anarchy and the price of stability for the accomplice game; details of this can be found in \cref{sec:appendix_accomplice}.

\paragraph{Number of $\NE$.} 

\cref{cor:NE_existence} showed that every accomplice manipulation game admits at least one pure strategy $\NE$. 
This raises the question about whether an accomplice game corresponding to an instance admits a unique NE. 
In fact, we show that the number of $\NE$ may indeed be factorially many. 

\begin{proposition}\label{prop:manyNEs}
Accomplice manipulation games may admit factorially many Nash equilibria.
\end{proposition}

\begin{proof}\label{ex:ManyNE}
Take a preference profile in which, for every $i<n$, $m_i$'s most preferred woman is $w_i$ and $w_i$'s most preferred man is $m_i$. Also, let $P=\{(m_n, w_n )\}$. Here, since the matching does not depend on  $m_n$'s ranking, every strategy profile (corresponding to different rankings submitted by $m_n$) is an $\NE$. So, there are $n!$ many NEs in this game.
\end{proof}

While in the example above all $\NE$ strategy profiles lead to the same matching, it is not always the case. For instance, in the example for the proof of \cref{thm:somestable} (ii), we show that an $\NE$ may correspond to an unstable matching but one resulting in a stable matching always exists. In fact, as we show in the next remark, we can also find $\NE$ profiles that correspond to different \emph{stable matchings}, and which can be obtained through different realizations of our push-up dynamics.

\begin{proposition}\label{prop:differentOutcomes}
For an accomplice manipulation game with the set of strategic pairs $P$, different realizations of the push-up dynamics may result in different $\NE$, each corresponding to a distinct stable matching. 
\end{proposition}

An example of such an instance was computer-generated and is provided in \cref{sec:appendix_accomplice}. This example also shows that, during the sequential misreports, some strategic pairs who did not have an incentive to misreport in the truthful profile may benefit from misreporting down the line.  

\section{Beyond Accomplice Manipulation Games 
} \label{sec:transferable}

We previously established the existence and computational tractability of finding an $\NE$ for the accomplice manipulation game through the construction a best-response dynamic (push-up dynamics) that converges to an $\NE$. In this section we show that our approach can be applied to simplify the analysis of other strategic games, specifically, one-for-many and woman self-manipulation game. 

\subsection{One-for-Many Manipulation Games} 
For every one-for-many manipulation game with strategic agents $P=(P_m, \{P_w^m\}_m)$,  we establish the following relation between the set of Nash equilibria of this game and an accomplice manipulation game with an appropriately defined strategic pairs $P'$ (where $P'\neq P$). 

\begin{restatable}{theorem}{oneforall}
    \label{th:one-for-all}
    For a one-for-many game with $(P_m, \{P_w^m\}_m)$ as the set of strategic players, consider an accomplice manipulation game in which the strategic players are defined as follows: $P' = \{(m,w)\, |\, m\in P_m \text{ and } w\in P_w^m \}$. Then, every $\NE$ profile $\succ_{\NE}$ of the accomplice game is also an $\NE$ profile of the one-for-many manipulation game.
\end{restatable}

So, to find an $\NE$ for the one-for-many game, it is enough to find an $\NE$ for its corresponding accomplice game. From Section \ref{sec:accomplice}, since this can be done through its push-up dynamics, the same response dynamics also gives the $\NE$ solution for the one-for-many game.

\begin{corollary}
    Consider a one-for-many manipulation game defined on the strategic men $P_m$ with $\{P^m_w\}_m$ as their corresponding beneficiaries. An $\NE$ profile for this game can be computed in polynomial time. Furthermore, at least one of the $\NE$ profiles results in a stable matching under $\DA$. 
\end{corollary}

\subsection{Woman Self-Manipulation Games}

As before, we construct a best-response dynamic which, when it converges, finds an $\NE$ of the woman self-manipulation game (see Appendix \ref{sec:appendix_sec5} for the definition). 

\paragraph{Inconspicuous Dynamics.} For strategic women $P_w$, an inconspicuous dynamic starting at $\stratprof^{(1)} =\, \succ$ is a sequence of strategy profiles $\overline{\stratprof}= (\stratprof^{(1)}, \ldots)$, where each subsequent profile $\stratprof^{(t+1)}$ {results from an optimal inconspicuous\footnote{A strategy $\stratprrof_w$ for a misreporting woman $w$ is said to be \textit{inconspicuous} to $\stratprof_w$ if it can be derived from $\stratprof_w$ by promoting at most one man and making no other changes to it.} self-manipulation by some strategic woman $w\in P_w$ in the matching instance $\langle M, W, \stratprof^{(t)}\rangle$, if possible.}  Else, $\stratprof^{(t+1)} = \stratprof^{(t)}$.

Notice the difference between this dynamics and the push-up dynamics for the accomplice game: The push-up dynamics consisted of no-regret, push-up manipulations, and we showed that finding such a strategy at each step $\stratprof^{(t)}$ was the same as finding a push-up manipulation in the instance $\langle M, W, \stratprof^{(t)} \rangle$. In contrast, the inconspicuous dynamic directly assumes that $\stratprof^{(t+1)}$ is obtained by finding an inconspicuous manipulation in $\langle M, W, \stratprof^{(t)} \rangle$. We show that the inconspicuous dynamic converges to an $\NE$. 

\begin{restatable}{theorem}{womanNE}
    \emph{($\NE$ of woman self-manipulation game)}\label{cor:NE_existence_woman}
For a woman manipulation game with strategic women $P_w\subseteq W$, the inconspicuous dynamics converges in polynomial time. Moreover, the converged point is an $\NE$.
\end{restatable}

The proof is analogous to the one for accomplice manipulation game, and is deferred to \cref{sec:appendix_sec5}.

\section{Experimental Results}
\label{sec:experiments}

\begin{figure*}[t] 
    \centering
    \begin{subfigure}{0.28\textwidth} 
        \centering
        \includegraphics[width=\linewidth]{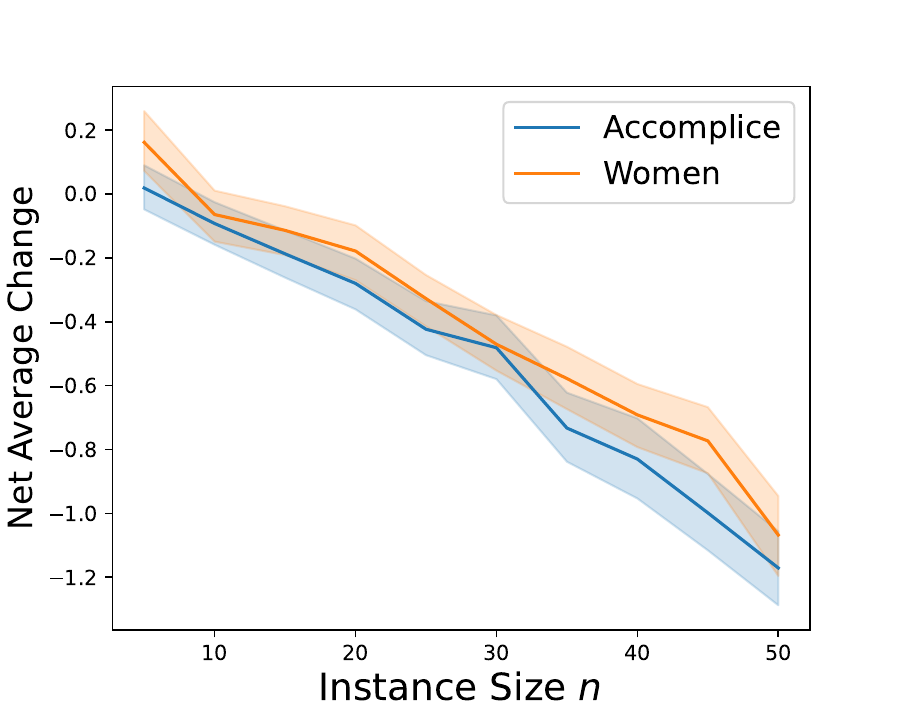} 
        \caption{Net welfare: Accomplice and woman self-manipulation games}
        \label{fig:welfare_acc_w}
    \end{subfigure}
    \hfill
    \begin{subfigure}{0.28\textwidth}
        \centering
        \includegraphics[width=1.2\linewidth]{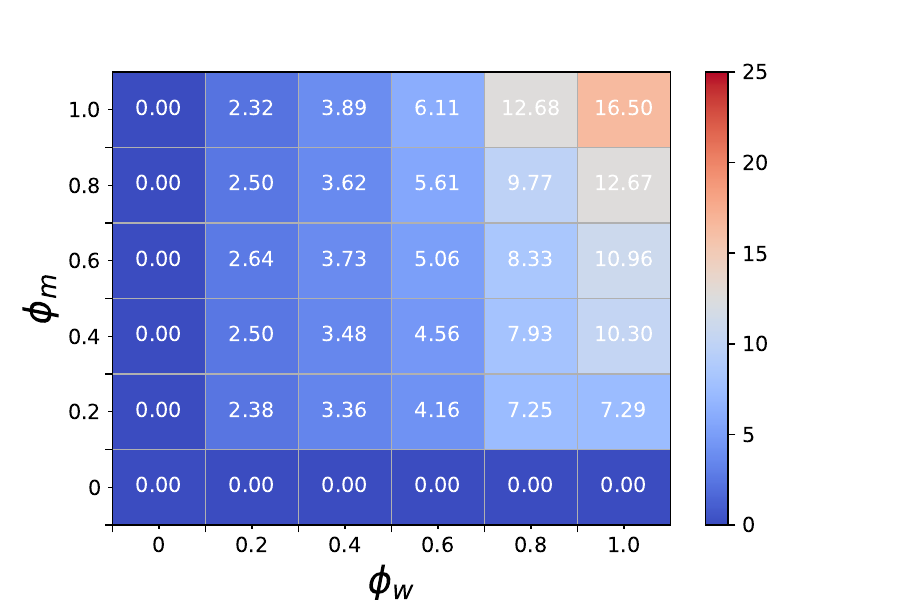}
        \caption{Accomplice game: Gain for best-off woman.}
        \label{fig:acc_best_off_phi}
    \end{subfigure}
    \hfill
  \begin{subfigure}{0.28\textwidth}
        \centering
        \includegraphics[width=1.2\linewidth]{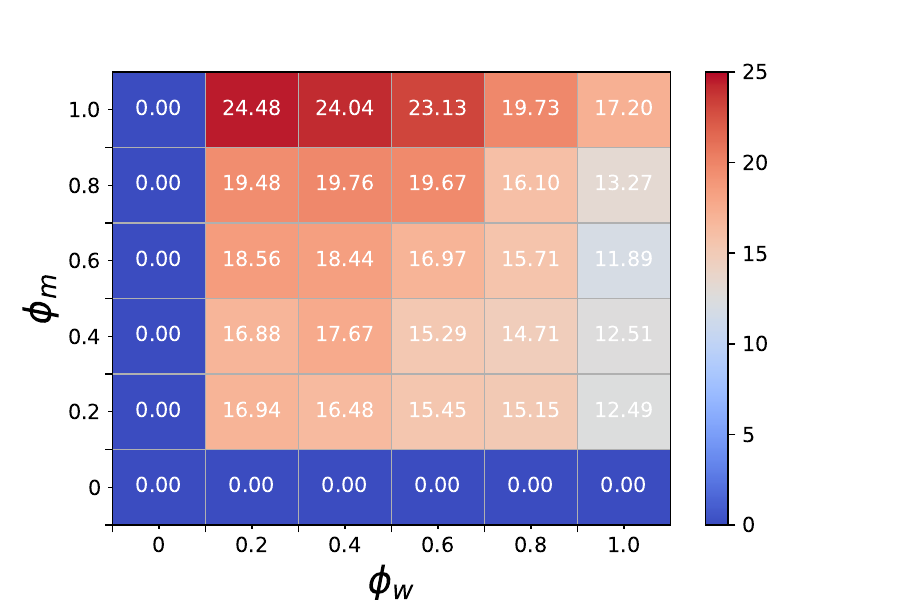}
        \caption{Self-manipulation game: Gain for best-off woman.}
        \label{fig:w_best_off_phi}
   \end{subfigure}
    \caption{Welfare of agents in $\NE$ of accomplice manipulation and woman self-manipulation games as a function of the instance size $n$ and the dispersion parameters $(\phi_w, \phi_m)$ in the Mallows' model.}
    \label{fig:benefit_phi}
\end{figure*}

While we theoretically showed the existence of a pure strategy $\NE$ and some of its properties, several questions remain open about its convergence in large-scale applications.
We investigate two properties:

\begin{enumerate}[(i)]
    \item \emph{Length of dynamics.} This is interesting since the length of these sequences can be viewed as a proxy for the complexity of the games, especially in practice. So, in reality, can agents reach an $\NE$ `naturally' through best-response actions? 
    \item \emph{Welfare of agents in an $\NE$.} Since  the constructed $\NE$ is stable and the set of stable matchings forms a lattice, women receive better partners in $\NE$ as compared to the truthful matching and men receive worse partners. But how large is their respective gain or loss?  
\end{enumerate}
For each of the above two properties, we study the effect of (a) the instance size, and (b) population as defined by correlation between preferences.
Instance sizes are measured by $n$, i.e., the number of men/women. We measure the correlation between preferences using \emph{Mallows' model} which is characterized by \emph{central ranking} $u$, and a \emph{dispersion parameter} $\phi \in [0,1]$ \cite{mallows1957non}. Here, the probability of selecting a ranking $v$ is proportional to $\phi^{\text{dist}(u,v)}$, where $\text{dist}(u,v)$ is the minimum number of pairwise swaps required to transform $u$ to $v$. The parameter $\phi$ determines the correlation between the preferences; $\phi=0$ corresponds to unanimity to $u$ and $\phi=1$ to impartial culture. Mallows' model has been previously studied in the context of matchings \citep{brilliantova2022fair,levy2017novel}. 

\paragraph{Length of Dynamics.}

To measure the impact of instance size, we generate 500 instances corresponding to each $n$ in $\{5, 10 , 15, \ldots, 50 \}$ using impartial culture.\footnote{Several works in stable matchings use IC to generate preferences \cite{ashlagi2017unbalanced,pittel1989average}.}
We observe that the average and the maximum length (across samples) (i) increase with $n$, (ii) are larger for the woman self-manipulation game than the accomplice game, and (iii) are relatively small even for large $n$ for both games, indicating that these games are not very complex. 

To evaluate the effect of correlated preferences, we generate our instances of size $n=30$,
and dispersion parameters $\phi_m, \phi_w$ (for men and women respectively) ranging $\{0, 0.2, \ldots, 1 \}$. For each pair $(\phi_m, \phi_w)$, we generate 500 instances by: (i) uniformly sampling central rankings $u_m$ and $u_w$ for men and women, and (ii) generating $n$ samples each from $u_m,u_w$ and their corresponding $\phi_m,\phi_w$ using Mallows' model. This generates a preference profile for both groups.

For both games, we see that (i) for $\phi_m=\phi_w$, the average length increases with $\phi_m$, and (ii) instances in which $\phi_m$ and $\phi_w$ are close to each other have larger average lengths. The plots of these experiments, and a similar evaluation of the  maximum length are provided in \cref{sec:appendix_sec6}.

\paragraph{Welfare of Agents in $\NE$.}

The instances for the two experiments (impact of instance size and impact of correlated preferences) are generated as before. The net gain (loss) for a woman (man) is measured as the increase (decrease) in the ranking of her (his) partner in the $\NE$ matching as compared to the truthful matching. 

We observe that, for both the games, the average gain for women and the average loss for men in $\NE$, as well as the gain for the best-off woman and the loss for the worst-off man, increases with $n$. Surprisingly, however, the loss of the men is larger than the gain for the women, i.e. the net welfare (for men and women together) is negative and decreases with $n$; see \cref{fig:welfare_acc_w}. Note that this observation addresses an important gap of how strategic behavior impacts welfare of agents in stable matching \citep{hosseini2024strategic}.

In our experiments for the impact of correlated preferences, we observe that the average gain for women and the gain for the best-off woman increase with the increase in either of the dispersion parameters for the accomplice manipulation game. Surprisingly, however, while the average gain in the woman self-manipulation game increases with $\phi_m$, it is the largest for when $\phi_w \sim \phi_m$. Moreover, we observe that the largest gain for the best-off woman occurs in instances with a high correlation in men's rankings but low correlation in women's preferences. Also, for smaller values of $(\phi_m,\phi_w)$, the average gain as well as the gain of the best-off woman is significantly higher for the woman self-manipulation game as compared to the accomplice manipulation game. 
These results are presented in Figure \ref{fig:benefit_phi} and in  \cref{sec:appendix_sec6}. 


\section{Conclusion and Future Work}

This work develops a technique that establishes the existence and polynomial time computation of accomplice manipulation and one-for-many manipulation games in two-sided matching markets. Specifically, we carefully construct a best-response dynamic---starting at the truthful profile---that converges to a Nash equilibrium. An interesting future direction would be to study the convergence of such dynamics starting from non-truthful profiles. 
%
Other future directions include understanding properties of mixed strategy $\NE$ in these games, as well as studying scenarios in which agents not only choose their misreports but also whether to engage in self-manipulation or manipulation through an accomplice.





\section*{Acknowledgments}
This project has received partial funding from the European Research Council
(ERC) under the European Union’s Horizon 2020 research and innovation
programme (grant agreement No 101002854), and was supported in part by NSF Awards IIS-2144413 and IIS-2107173.

\begin{center}
  \includegraphics[width=4cm]{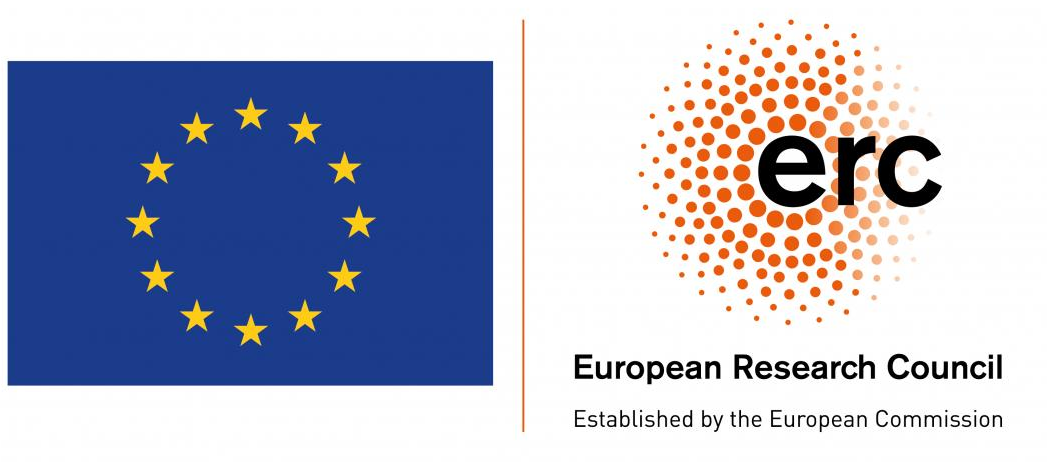}
  \hspace{3em}
  \includegraphics[width=3.5cm]{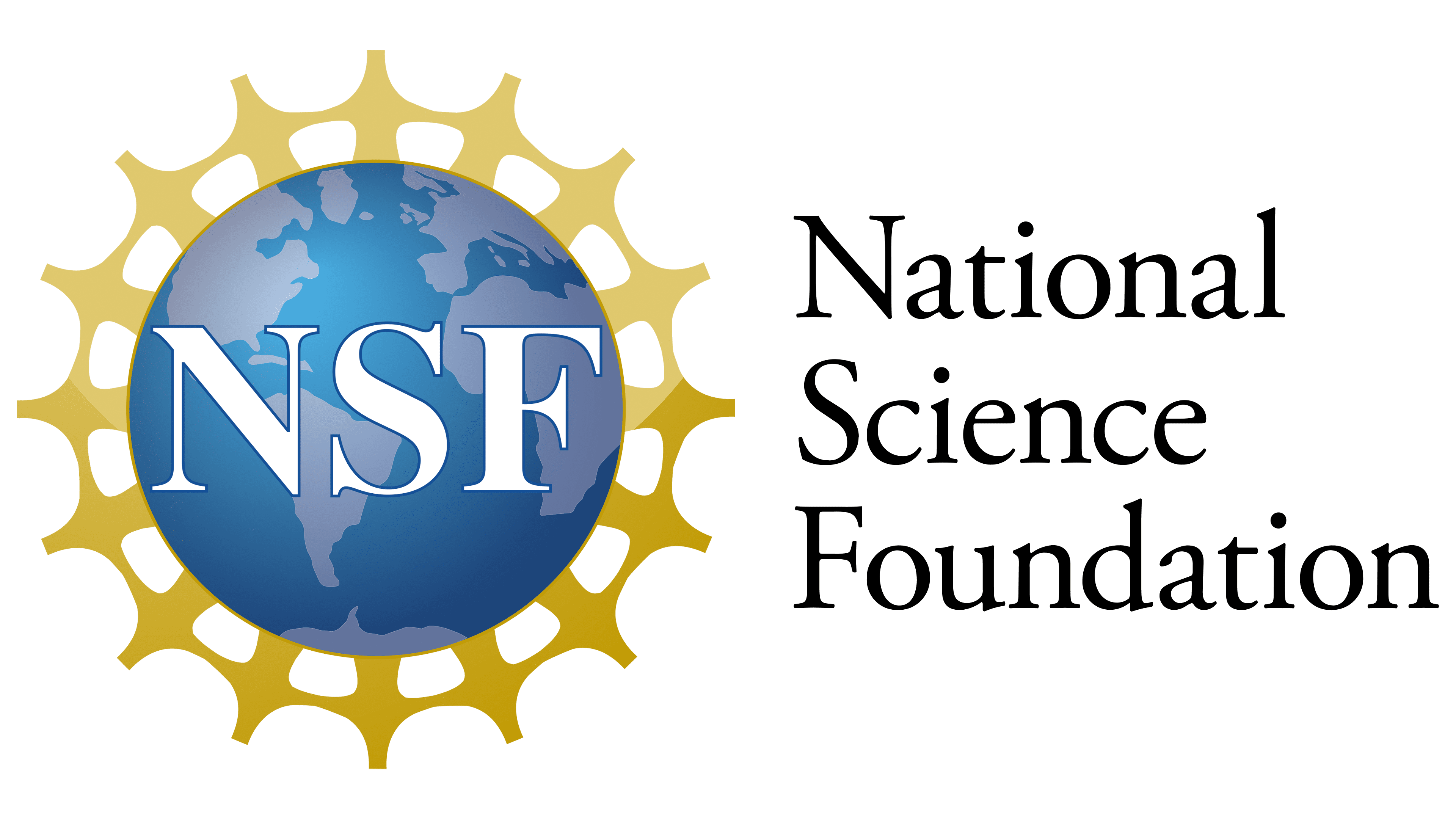}
\end{center}

\bibliographystyle{named}
\bibliography{ref}
\newpage

\appendix

\section{Omitted Material from Section \ref{sec:notation}} \label{sec:appendix_figures}





Here, we illustrate the notion of a push-up operation. The axis in the figure represents the preference list of the man $m$ in the profile $\stratprof$, while $\mu_\stratprof(m)$ is his match under profile $\stratprof$. Hence, $\mathcal{R}(\stratprof_m, \mu_\stratprof)$ is the set of women that are less preferred than $\mu_\stratprof(m)$ in $\stratprof_m$, and symmetrically $\mathcal{L}(\stratprof_m, \mu_\stratprof)$ is the set of those women who are more preferred than $\mu_\stratprof(m)$ in $\stratprof_m$. Then, the man $m$ revises his preference list by shifting a set of women $X$ above his match $\mu_\stratprof(m)$ in the profile $\stratprof$.

 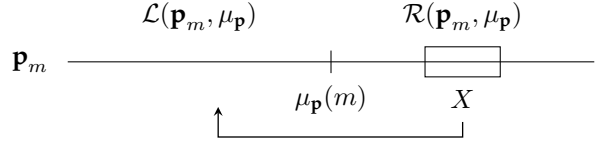
\begin{figure}[h!]
    \centering
     \begin{tikzpicture}
         \draw (0, 0) -- (7, 0);
        
         \node at (1.75, 0.6) {$\mathcal{L}(\stratprof_m, \mu_\stratprof)$};
        
         \node at (5.25, 0.6) {$\mathcal{R}(\stratprof_m, \mu_\stratprof)$};
        
         \node at (3.5, -0.5) {$\mu_\stratprof(m)$};
        
         \node at (-0.5, 0) {$\stratprof_m$};

         \draw (3.5, -0.15) -- (3.5, 0.15);
        
         \draw (4.75, -0.2) rectangle (5.75, 0.2);
         \node at (5.25, -0.5) {$X$};

         \draw[->, >=stealth, line width=0.5pt] (5.25, -0.8) -- (5.25, -1) -- (2, -1) -- (2, -0.6);
     \end{tikzpicture}
     \caption{Illustrating push-up operations by man $m$ in strategy profile $\stratprof$.}
     \label{fig:push-up_1}
 \end{figure}

In fact, $\mathcal{L}(\stratprof_m, \mu_\stratprof)$ and $\mathcal{R}(\stratprof_m, \mu_\stratprof)$ can generalized to $\mathcal{L}(\stratprrof_m, \mu_\stratprof)$ and $\mathcal{R}(\stratprrof_m, \mu_\stratprof)$. Here, the axis represents the preference list of $m$ in profile $\stratprrof$. As before, $\mu_\stratprof(m)$ is $m$'s match under profile $\stratprof$, and $\mathcal{R}(\stratprrof_m, \mu_\stratprof)$ is the set of less-preferred women than  $\mu_\stratprof(m)$ in the profile $\stratprrof_m$ and $\mathcal{L}(\stratprrof_m, \mu_\stratprof)$ is the set of those women who are more preferred than $\mu_\stratprof(m)$ in $\stratprrof_m$.
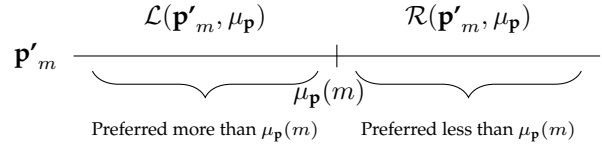
\begin{figure}[h!]
    \centering
     \begin{tikzpicture}
         \draw (0, 0) -- (7, 0);
        
         \draw [decorate,decoration={brace,amplitude=10pt,mirror}] (0.25, -0.2) -- (3.25, -0.2);
         \node at (1.75, 0.5) {$\mathcal{L}(\stratprrof_m, \mu_\stratprof)$};
         \node at (1.75, -1) {$\text{\scriptsize Preferred more than $\mu_\stratprof(m)$}$};
        
         \draw [decorate,decoration={brace,amplitude=10pt,mirror}] (3.75, -0.2) -- (6.75, -0.2);
         \node at (5.25, 0.5) {$\mathcal{R}(\stratprrof_m, \mu_\stratprof)$};
         \node at (5.25, -1) {$\text{\scriptsize Preferred less than $\mu_\stratprof(m)$}$};
        
         \node at (3.4, -0.5) {$\mu_\stratprof(m)$};
        
         \node at (-0.5, 0) {$\stratprrof_m$};

         \draw (3.5, -0.15) -- (3.5, 0.15);
     \end{tikzpicture}
     \caption{Illustrating the notation $\mathcal{L}(\stratprrof_m, \mu_\stratprof)$ and $\mathcal{R}(\stratprrof_m, \mu_\stratprof)$.}
     \label{fig:L_R}
 \end{figure}

\section{Omitted Material from Section \ref{sec:accomplice}} \label{sec:appendix_accomplice}

\subsection{From Section \ref{sec:NE}}

\begin{proposition}\label{prop:GS_lattice}[\cite{gale1962college,mcvitie1971stable}] 
    Given any preference profile $\succ$, let $\mu_\succ \coloneqq\DA(\succ)$. Then, for any stable matching $\mu'\in \mathcal{S}_\succ$,  $\mu_\succ \succeq_M {\mu'}$ and ${\mu'} \succeq_W \mu_\succ $. 
\end{proposition}


\lemNEsubsets*

\begin{proof}
First, notice that since the set of manipulating men in $P$ and $P'$ is the same, every strategy profile in $P$ is also a strategy profile in $P'$, and vice versa.

Now, towards a contradiction, suppose that there is a strategy profile $\stratprof$ that is a Nash equilibrium in the accomplice game with strategic pairs $P'$ but is not a Nash equilibrium in the game with strategic pairs $P$. This means that there exists a pair $(m,w)\in P$ that can perform a valid accomplice manipulation (let the manipulated profile be $\stratprrof$) at $\stratprof$. Since $P_w\subseteq P'_w$, the pair $(m,w)$ is also a strategic pair in the game defined by the pairs $P'$, i.e., $(m,w) \in P'$. However, this means that $\stratprrof$ is a valid accomplice manipulation in the game defined by the strategic pairs $P'$ as well; this contradicts our assumption and completes the proof.

\end{proof}

\lemshrinkingset*

\begin{proof}
    The proof is via induction; in fact we will prove a stronger containment for (ii): $\mathcal{R}(\stratprof^{(t+1)}_m, w)\subseteq \mathcal{R}(\stratprof^{(t)}_m, w) \subseteq \ldots \subseteq \mathcal{R}(\succ_m, w)$ for every $w\in \{\mu_{\stratprof^(t+1)}(m)\}\cup \mathcal{R}(\stratprof^{(t+1)}_m, \mu_{\stratprof^(t+1)})$. Consider the base case ($i=1$) in which the consecutive strategy profiles are $\stratprof^{(1)} =\, \succ$ and $\stratprof^{(2)}$ which results from a push-up accomplice manipulation by a pair $(m',w')$ on $\succ$. Since it is a no-regret, push-up accomplice manipulation, $m'$ does not become worse-off, i.e., $\mu_{\stratprof^{(2)}} \succeq_{m'} \mu_{\stratprof^{(1)}}$. Thus, for the manipulator $m'$, we have $\mathcal{R}(\stratprof^{(2)}_{m'}, \mu_{\stratprof^{(2)}}) \subseteq \mathcal{R}(\stratprof^{(1)}_{m'}, \mu_{\stratprof^{(2)}})$. For any woman $w\in \mu_{\stratprof^{(2)}}$, we also know that $\mathcal{R}(\stratprof^{(2)}_{m'}, w) \subseteq \mathcal{R}(\stratprof^{(1)}_{m'}, w)$ for every woman $w\in \mathcal{R}(\stratprof^{(2)}_{m'}, \mu_{\stratprof^{(2)}})$ since the manipulation was through a push-up operation. Then, from Proposition \ref{prop:accomplice_properties}, we have $\mathcal{S}_{\stratprof^{(2)}} \subseteq \mathcal{S}_{\stratprof^{(1)}}$.  
    Since the preference list of any man $m\neq m'$ has not changed, we have $\mathcal{R}(\stratprof^{(2)}_m, w) = \mathcal{R}(\stratprof^{(1)}_m, w)$. This shows that the set inclusions hold for the base case.

    Next, assume that 
    \begin{enumerate}[(i)]
        \item $\mathcal{S}_{\stratprof^{(t)}} \subseteq \mathcal{S}_{\stratprof^{(t-1)}}\subseteq \ldots \subseteq \mathcal{S}_{\stratprof^{(1)}} = \mathcal{S}_{\succ}$, and
        \item $\mathcal{R}(\stratprof^{(t)}_m, w)\subseteq \mathcal{R}(\stratprof^{(t-1)}_m, w) \subseteq \ldots \subseteq \mathcal{R}(\succ_m, w)$ for every woman $w\in \{\mu_{\stratprof^{(t)}}\}\cup \mathcal{R}(\stratprof^{(t)}_m, \mu_{\stratprof^{(t)}})$
    \end{enumerate} and consider a push-up accomplice manipulation by $(m'', w'')$ at $\stratprof^{(t)}$, which leads to the profile $\stratprof^{(t+1)}$.
    For the manipulator man $m''$, since the manipulation is no-regret and via push-up operations at $\stratprof^{(t)}$, we have that $\mu_{\stratprof^{(t+1)}} \succeq_{m''} \mu_{\stratprof^{(t)}}$, i.e., $\mu_{\stratprof^{(t+1)}} \notin \mathcal{R}(\stratprof^{(1)}_{m''}, \mu_{\stratprof^{(t)}})$. By the set inclusions (ii) in the inductive argument, we know that $\mu_{\stratprof^{(t+1)}} \notin \mathcal{R}(\stratprof^{(t)}_{m''}, \mu_{\stratprof^{(t)}})$, i.e., $\mu_{\stratprof^{(t+1)}} \succeq^{\stratprof^{(t)}}_{m''} \mu_{\stratprof^{(t)}}$. 
    Since this was a no-regret push-up manipulation at $\stratprof^{(t)}$, by applying Proposition \ref{prop:accomplice_properties} to the stable matching instance $\langle M, W, \stratprof^{(t)} \rangle$, we get that $\mathcal{S}_{\stratprof^{(t+1)}} \subseteq \mathcal{S}_{\stratprof^{(t)}}$, and from the inductive argument, we have $\mathcal{S}_{\stratprof^{(t+1)}} \subseteq \mathcal{S}_{\stratprof^{(t)}} \subseteq \ldots \subseteq \mathcal{S}_{\succ}$; this proves part (i).

    For part (ii), from \Cref{prop:accomplice_properties}, we know $\mu_{\stratprof^{(t+1)}}\in \mathcal{S}_{\stratprof^{(t)}}$; thus $\mu_{\stratprof^{(t)}} \succ_M^{\stratprof^{(t)}} \mu_{\stratprof^{(t+1)}}$, i.e., $\mu_{\stratprof^{(t+1)}}\in \mathcal{R}(\stratprof_m^{(t)} , \mu_{p^{(t)}})$ [whereby $\mathcal{R}(\stratprof^{(t+1)}_m, \mu_{\stratprof^{(t+1)}})\subseteq\mathcal{R}(\stratprof^{(t)}_m, \mu_{\stratprof^{(t+1)}})$ and $\mathcal{R}(\stratprof^{(t+1)}_m, w)\subseteq\mathcal{R}(\stratprof^{(t)}_m, w\in\mathcal{R}(\stratprof^{(t+1)}_m, w))$ for every woman $w\in $]. Thus, from the induction hypothesis, $\mathcal{R}(\stratprof^{(t)}_m, \mu_{\stratprof^{(t+1)}})\subseteq \mathcal{R}(\stratprof^{(t-1)}_m, \mu_{\stratprof^{(t+1)}}) \subseteq \ldots \subseteq \mathcal{R}(\succ_m, \mu_{\stratprof^{(t+1)}})$. Moreover, every woman $w\in \mathcal{R}(\stratprof^{(t+1)}_m, \mu_{\stratprof^{(t+1)}})$ also belongs to $\mathcal{R}(\stratprof^{(t)}_m, \mu_{\stratprof^{(t+1)}})$. This completes the proof.

\end{proof}

For the manipulator $m''$, no regret so $\mu_{\stratprof^{t+1}} = \mu_{\stratprof^{t}}$. Thus, since push-up, we know $\mathcal{R}(\stratprof_m^{(t+1)}, \mu_{\stratprof^{t+1}}) \subseteq ...$

\actualiscurrent*

\begin{proof}
Towards a contradiction, assume that $\mu_{\stratprof^{(t)}} \succ^{\stratprof^{(t)}}_m \mu_{\stratprof^{(t+1)}}$ for some man $m$. This means that $\mu_{\stratprof^{(t+1)}}(m) \in \mathcal{R}(\stratprof^{(t+1)}_m, \mu_{\stratprof^{(t)}})$. By the set inclusions in Lemma \ref{lem:shrinking set} (ii), we get that $\mu_{\stratprof^{(t+1)}}(m) \in\, \mathcal{R}(\succ_m, \mu_{\stratprof^{(t)}})$, i.e., $\mu_{\stratprof^{(t)}} \succ_m \mu_{\stratprof^{(t+1)}}$. But this contradicts the construction of the push-up dynamics, which ensures that the manipulation is no-regret. Specifically, it contradicts condition (ii) which states that $\mu_{\stratprof^{(t+1)}} \succeq_{m} \mu_{\stratprof^{(t)}}$. 
\end{proof}

For completing the proof of Theorem \ref{cor:NE_existence}, we need the following final proposition due to \citet{huang2006cheating}.

\begin{proposition}\label{prop:permute}\citep{huang2006cheating}. For a preference profile $\succ$ and its corresponding \DA{} matching $\mu_\succ$, consider $\succ'_m := (\pi(\mathcal{L}(\succ_m,\mu_\succ)) ,\mu_\succ(m), \pi(\mathcal{R}(\succ_m, \mu_m)))$ where $\pi(\mathcal{L}(\succ_m,\mu_\succ))$ and $\pi(\mathcal{R}(\succ_m,\mu_\succ))$ are arbitrary permutations of $\mathcal{L}(\succ_m,\mu_\succ)$ and $\mathcal{R}(\succ_m,\mu_\succ))$ respectively. Then for $\succ' := (\succ_{-m}, \succ'_m)$ and $\mu_{\succ'}:= \DA(\succ')$, we have that $\mu_\succ = \mu_{\succ'}$.
\end{proposition}

We are now ready to prove Theorem \ref{cor:NE_existence}.

\paragraph{Proof of Theorem \ref{cor:NE_existence}}

\begin{proof}


Here, we describe the remaining part of the proof, i.e., a polynomial-time procedure for checking if $\stratprof$ is an $\NE$. Towards this, consider the stable matching instance $\langle M, W, \stratprof \rangle$, and recall from Proposition \ref{prop:accomplice_properties}, if there exists an accomplice pair $(m,w)$ in this instance, then $(m,w)$ can manipulate inconspicuously with respect to $\stratprof$. For each pair $(m,w) \in P$, consider every possible inconspicuous operation, i.e.,  all scenarios in which different women from $\mathcal{R}(\stratprof_m, \mu_\stratprof)$ are shifted to the first place in $m$'s ranking (only pushing-up to first place suffices due to Proposition \ref{prop:permute}). There are atmost $O(n)$ such modified instances, and by running $\DA$ on every such modified instance, we can check in polynomial time if $w$ benefits from such an operation and if $m$ incurs any regret with respect to the true preference $\succ_m$, i.e, check if the modified instance is an accomplice manipulation. If such an inconspicuous operation exists for some $(m,w)\in P$, then we conclude that $\stratprof$ is not an $\NE$; otherwise, since $\stratprof_w =\, \succ_w$, by definition $\stratprof$ is an $\NE$. Note that our procedure only requires $O(n|P_m|)$ runs of the $\DA$ algorithm.
Thus, the verification procedure runs in polynomial time. 
\end{proof}

\subsection{From Section \ref{sec:properties_NE}}

The loss-in-stability due to strategic behavior can be measured in terms of the \emph{Price of Anarchy} and \emph{Price of Stability}, which are defined in terms of the number of stable pairs in a matching $\mu$, denoted $\NSP(\mu)$. A pair $(m,w)$ is said to be stable in a matching $\mu$ if it does not block $\mu$.

\begin{definition}[Price of Anarchy and Price of Stability]
    In an accomplice manipulation game with $P \subseteq M\times W$ as the strategic pairs, the Price of Anarchy ($\PoA$) and the Price of Stability ($\PoS$) are given by 
    \begin{align*}
        \PoA = \frac{\max_{\text{any $\mu$}}\NSP(\mu)}{\min_{\text{$\mu\in\NE$}}\NSP(\mu)} = \frac{n^2}{\min_{\text{$\mu\in\NE$}}\NSP(\mu)},\\
        \PoS = \frac{\max_{\text{any $\mu$}}\NSP(\mu)}{\max_{\text{$\mu\in\NE$}}\NSP(\mu)} = \frac{n^2}{\max_{\text{$\mu\in\NE$}}\NSP(\mu)}.
    \end{align*} Here, $\NSP(\mu)$ denotes the number of stable (or non-blocking) pairs in matching $\mu$.
\end{definition}

Next, we determine the Price of Anarchy and Price of Stability of the accomplice manipulation games.

\begin{theorem}[$\PoA$ and $\PoS$ of accomplice manipulation games] \label{cor:PoA_PoS}
    For the accomplice manipulation game defined by the strategic pairs $P\subseteq M\times W$, the Price of Anarchy and Price of stability are $\PoA \leq \frac{n^2}{|P|}$ and $ \PoS = 1$ respectively.
\end{theorem}

\begin{proof}
    As a corollary of Theorem \ref{thm:somestable} we observe that the price of stability is $ \PoS = 1$. 
    This is because Part 1 of Theorem \ref{thm:somestable} shows that a Nash equilibrium which corresponds to a stable matching can always be found, even though all Nash equilibria may not correspond to stable matchings. Thus, the best-case guarantee, i.e., $\PoS=1$.

    For the price of anarchy, Theorem \ref{thm:somestable}(ii) provides the following upper bound: $\PoA \leq \frac{n^2}{|P|}$. Through the following example, we show that in some cases $\PoA = \frac{n^2}{|P|}$.

    Consider a stable matching instance consisting of five men and five women, with preferences as listed below. The $\DA$ matching is underlined in the instance.

    \begin{table}[ht]
\centering 
    \begin{tabularx}{0.9\linewidth}{XXXXXXXXXXXXXXXXXXXXXXX}
    
         	$m_1$: & $w_2^*$& \underline{$w_1$} & $w_3$ & $w_4$ & $w_5$ && $w_1$: & \underline{$m_1$} & $m_3$ & $m_2^*$ & $m_4$ & $m_5$\\
            $m_2$: & {$w_1^{*}$} & \underline{$w_2$} & $w_3$ & $w_4$ & $w_5$ && $w_2$: & \underline{$m_2$} & $m_3$ &{$m_1^{*}$} & $m_4$ & $m_5$\\
            $m_3$: &{$w_1$} & $w_2$ & \underline{$w_3^*$} & $w_4$ & $w_5$ && $w_3$: & \underline{$m_3^*$} & $m_1$ & $m_2$ & $m_4$ & {$m_5$}\\
            $m_4$: & \underline{$w_4$} & $w_5^*$ & $w_1$ & $w_2$ & $w_3$ && $w_4$: & {$m_5^{*}$} & $m_3$ & $m_1$ & $m_2$ & \underline{$m_4$}\\
            $m_5$: & \underline{$w_5$} & $w_4^*$ & $w_1$ & $w_2$ & $w_3$ && $w_5$: & $m_4^*$ & $m_1$ & $m_2$ & {$m_3$} & \underline{$m_5$}
    \end{tabularx}
\vspace{-0.06in}
\end{table}

Let the set of strategic pairs be given by $P= M\times W \backslash\{(m_3, w_1), (m_3, w_2)\}$, i.e., all possible pairs apart from $(m_3,w_1)$ and $(m_3, w_2)$ are strategic pairs. Consider the strategic profile $\succ'$ where $m_3$ reports $w_4\succ' w_3\succ'w_1 \succ' w_2 \succ' w_5$ and everyone else reports truthfully. One can verify that this is a Nash equilibrium profile and its corresponding matching $\mu_\succ' = \DA(\succ')$ is the one marked by $^*$. However, both $(m_3, w_1)$ and $(m_3, w_2)$ are blocking pairs. Thus, the number of stable pairs in the matching $^*$ is $\NSP(^*) = n^2-2 = |P|$; thereby, $\PoA = \frac{n^2}{|P|}$. 

\end{proof}


\paragraph{Example of an instance for which different realizations of the push-up dynamics leads to different $\NE$ matchings.} 

Here's an instance with $n=30$ and the set of strategic pairs as $P=M\times W$
Both, men and women, are named as integers in $[1,30]$. The preference lists of individual agents are represented as a list of agents' from the other side, ordered from most to least preferred. Below, are the true preferences of the agents.

\vspace{0.5em}

Men's preference profile: [[22, 18, 15, 28, 25, 3, 19, 1, 23, 20, 16, 8, 27, 2, 24, 30, 21, 9, 10, 13, 11, 14, 26, 17, 29, 12, 6, 7, 5, 4], [25, 1, 4, 23, 3, 15, 22, 12, 26, 8, 7, 17, 9, 11, 16, 28, 19, 10, 2, 14, 5, 18, 21, 29, 30, 27, 24, 6, 20, 13], [18, 16, 12, 13, 22, 27, 15, 30, 4, 2, 11, 23, 1, 28, 3, 24, 5, 26, 19, 29, 17, 8, 20, 10, 21, 6, 25, 9, 7, 14], [11, 30, 19, 23, 16, 14, 17, 21, 26, 10, 28, 27, 15, 13, 8, 5, 29, 4, 22, 6, 20, 7, 1, 25, 3, 18, 12, 24, 9, 2], [30, 16, 19, 7, 2, 27, 24, 13, 8, 6, 17, 20, 12, 9, 5, 23, 3, 1, 22, 25, 26, 4, 15, 11, 10, 18, 21, 29, 14, 28], [18, 17, 24, 7, 25, 29, 27, 28, 14, 10, 5, 19, 9, 4, 6, 20, 23, 3, 13, 16, 2, 12, 1, 26, 30, 15, 11, 22, 8, 21], [28, 6, 7, 18, 3, 10, 17, 21, 19, 22, 5, 12, 30, 9, 14, 13, 20, 1, 23, 24, 27, 11, 29, 26, 25, 4, 15, 2, 8, 16], [15, 7, 14, 5, 28, 27, 20, 2, 22, 1, 6, 18, 24, 29, 21, 3, 9, 17, 23, 19, 4, 10, 30, 25, 26, 13, 8, 12, 11, 16], [14, 4, 19, 17, 29, 24, 30, 8, 12, 27, 2, 6, 9, 1, 20, 15, 3, 16, 18, 11, 13, 22, 21, 25, 23, 5, 26, 7, 28, 10], [7, 28, 2, 19, 22, 4, 29, 24, 15, 17, 25, 23, 5, 6, 1, 21, 13, 11, 10, 8, 30, 9, 16, 20, 12, 3, 18, 14, 26, 27], [8, 22, 28, 10, 1, 9, 16, 4, 14, 6, 17, 19, 24, 25, 5, 18, 13, 11, 21, 3, 30, 27, 26, 2, 23, 29, 15, 20, 7, 12], [7, 29, 26, 9, 5, 2, 1, 15, 12, 8, 3, 25, 16, 24, 17, 28, 14, 19, 23, 10, 4, 13, 11, 6, 18, 30, 27, 22, 21, 20], [21, 13, 15, 8, 2, 24, 6, 12, 30, 4, 23, 7, 17, 10, 27, 20, 26, 14, 29, 19, 9, 11, 25, 18, 1, 22, 16, 5, 3, 28], [20, 17, 23, 30, 27, 10, 3, 16, 24, 5, 25, 13, 12, 4, 1, 28, 29, 21, 22, 2, 7, 14, 8, 6, 11, 18, 19, 15, 9, 26], [15, 25, 9, 12, 26, 20, 8, 23, 28, 11, 14, 18, 22, 5, 17, 30, 27, 3, 29, 7, 1, 13, 10, 21, 2, 24, 16, 6, 19, 4], [16, 19, 12, 24, 5, 1, 3, 10, 22, 2, 26, 23, 14, 25, 11, 9, 15, 28, 4, 30, 7, 6, 29, 27, 8, 13, 17, 21, 18, 20], [15, 9, 1, 29, 24, 12, 8, 7, 25, 4, 23, 6, 5, 22, 13, 17, 26, 21, 30, 11, 16, 28, 18, 19, 10, 2, 3, 14, 27, 20], [21, 17, 9, 20, 1, 15, 6, 23, 11, 22, 29, 3, 12, 8, 24, 5, 7, 4, 2, 26, 16, 19, 18, 27, 28, 14, 13, 30, 25, 10], [24, 30, 16, 29, 15, 22, 23, 7, 28, 26, 14, 19, 1, 9, 11, 18, 27, 10, 17, 25, 6, 12, 2, 8, 3, 4, 20, 5, 13, 21], [25, 29, 1, 23, 12, 27, 7, 2, 19, 30, 18, 22, 5, 13, 3, 26, 16, 4, 28, 20, 9, 17, 15, 21, 10, 11, 8, 6, 14, 24], [21, 2, 25, 24, 26, 14, 22, 28, 20, 18, 9, 12, 8, 5, 17, 16, 23, 11, 30, 4, 15, 13, 19, 10, 7, 3, 29, 6, 1, 27], [17, 7, 10, 29, 4, 26, 9, 5, 14, 22, 1, 18, 3, 25, 20, 28, 2, 30, 11, 13, 6, 19, 23, 12, 8, 15, 24, 21, 16, 27], [27, 28, 29, 25, 20, 8, 7, 24, 30, 12, 17, 11, 13, 5, 22, 2, 3, 23, 26, 19, 4, 18, 10, 16, 1, 15, 21, 9, 14, 6], [21, 29, 6, 25, 4, 12, 23, 20, 7, 22, 11, 19, 5, 2, 24, 9, 18, 26, 16, 3, 28, 8, 27, 17, 10, 30, 1, 14, 15, 13], [2, 30, 22, 3, 13, 1, 11, 19, 21, 17, 29, 8, 23, 26, 14, 10, 28, 6, 4, 15, 7, 16, 25, 12, 9, 24, 18, 5, 20, 27], [20, 21, 7, 15, 2, 19, 25, 16, 29, 17, 13, 23, 10, 22, 30, 12, 18, 14, 27, 26, 3, 4, 6, 28, 11, 5, 8, 24, 9, 1], [6, 7, 22, 13, 5, 18, 16, 17, 4, 3, 11, 12, 9, 27, 24, 30, 25, 1, 15, 10, 26, 28, 23, 20, 29, 14, 19, 2, 8, 21], [6, 13, 28, 8, 26, 10, 9, 25, 23, 15, 2, 29, 16, 3, 17, 5, 24, 7, 1, 11, 14, 20, 27, 19, 21, 18, 22, 4, 12, 30], [29, 17, 11, 30, 10, 13, 24, 21, 19, 25, 23, 1, 12, 27, 6, 3, 9, 4, 8, 14, 26, 18, 22, 2, 20, 28, 5, 15, 7, 16], [2, 18, 16, 28, 5, 24, 9, 15, 30, 13, 29, 17, 22, 11, 7, 10, 8, 19, 23, 27, 21, 20, 14, 3, 25, 1, 26, 6, 4, 12]]

\vspace{0.5em}

Women's preference profile: [[21, 5, 14, 23, 4, 1, 27, 24, 13, 17, 11, 22, 2, 10, 19, 8, 3, 15, 16, 18, 6, 26, 20, 12, 29, 7, 9, 28, 25, 30], [14, 11, 23, 26, 4, 21, 10, 28, 16, 3, 25, 5, 17, 20, 1, 27, 24, 9, 29, 12, 6, 30, 7, 13, 22, 15, 19, 8, 18, 2], [4, 29, 27, 23, 28, 17, 8, 26, 5, 25, 6, 16, 14, 15, 1, 24, 10, 7, 13, 11, 3, 18, 22, 30, 9, 12, 2, 21, 19, 20], [10, 13, 12, 19, 4, 11, 24, 28, 6, 23, 3, 25, 14, 17, 20, 5, 27, 15, 26, 8, 16, 18, 2, 1, 21, 7, 29, 30, 9, 22], [17, 29, 20, 9, 10, 13, 26, 23, 22, 28, 4, 5, 21, 30, 3, 19, 27, 12, 24, 2, 7, 18, 16, 25, 1, 8, 14, 6, 11, 15], [10, 22, 5, 1, 7, 20, 2, 30, 4, 14, 8, 26, 16, 12, 3, 6, 9, 18, 21, 13, 29, 28, 23, 19, 11, 24, 27, 15, 25, 17], [26, 13, 16, 8, 24, 25, 20, 4, 7, 5, 2, 14, 9, 29, 27, 18, 15, 19, 11, 1, 22, 30, 17, 12, 23, 3, 10, 6, 21, 28], [1, 28, 15, 11, 10, 25, 8, 30, 21, 26, 16, 19, 27, 17, 23, 12, 24, 5, 20, 2, 18, 7, 6, 14, 13, 3, 22, 9, 4, 29], [2, 27, 4, 11, 20, 6, 22, 7, 16, 5, 17, 24, 14, 26, 15, 10, 30, 13, 19, 23, 3, 21, 1, 25, 18, 28, 8, 9, 29, 12], [30, 4, 7, 28, 3, 10, 13, 27, 20, 16, 14, 26, 5, 2, 11, 9, 24, 23, 6, 17, 19, 1, 15, 22, 25, 8, 29, 12, 18, 21], [29, 7, 28, 23, 26, 10, 20, 21, 24, 6, 1, 18, 15, 19, 5, 11, 13, 16, 4, 17, 3, 8, 12, 14, 25, 9, 27, 22, 30, 2], [4, 28, 9, 27, 8, 30, 1, 25, 24, 6, 2, 16, 19, 29, 15, 13, 5, 12, 21, 3, 17, 7, 26, 11, 10, 18, 20, 23, 22, 14], [18, 28, 5, 7, 23, 2, 24, 22, 16, 13, 10, 9, 21, 1, 14, 29, 3, 19, 6, 25, 4, 27, 17, 20, 30, 12, 8, 26, 15, 11], [11, 14, 7, 16, 26, 19, 5, 21, 10, 12, 15, 4, 1, 22, 28, 30, 13, 17, 8, 6, 27, 2, 24, 20, 25, 3, 18, 9, 29, 23], [23, 9, 8, 21, 22, 29, 24, 15, 14, 27, 4, 13, 10, 25, 11, 7, 28, 26, 1, 5, 16, 12, 2, 20, 6, 3, 18, 30, 17, 19], [16, 29, 6, 5, 26, 2, 7, 3, 24, 22, 19, 8, 13, 18, 12, 15, 17, 9, 10, 25, 20, 11, 30, 28, 27, 4, 21, 1, 23, 14], [21, 23, 2, 13, 28, 22, 27, 10, 7, 29, 8, 9, 25, 11, 24, 6, 15, 30, 17, 19, 26, 20, 18, 16, 3, 12, 4, 14, 1, 5], [18, 24, 30, 10, 1, 14, 26, 16, 4, 7, 29, 25, 8, 5, 19, 13, 9, 3, 23, 2, 21, 27, 17, 20, 12, 15, 6, 11, 28, 22], [4, 19, 12, 29, 28, 30, 9, 3, 24, 7, 13, 18, 20, 15, 5, 23, 17, 16, 25, 6, 11, 21, 26, 2, 22, 1, 10, 8, 14, 27], [14, 11, 10, 28, 25, 18, 26, 15, 13, 21, 1, 30, 17, 6, 2, 23, 12, 5, 16, 20, 4, 19, 22, 8, 7, 29, 24, 9, 3, 27], [25, 18, 19, 9, 30, 28, 27, 17, 26, 7, 11, 13, 24, 29, 22, 21, 6, 12, 1, 5, 2, 23, 15, 4, 20, 10, 16, 3, 14, 8], [7, 13, 24, 19, 9, 3, 15, 17, 10, 2, 12, 8, 14, 22, 18, 27, 23, 21, 29, 25, 28, 4, 26, 16, 6, 11, 1, 30, 5, 20], [7, 1, 27, 14, 29, 4, 13, 10, 11, 12, 9, 8, 21, 24, 23, 5, 19, 17, 22, 26, 28, 16, 25, 3, 15, 18, 20, 30, 2, 6], [4, 25, 21, 10, 29, 16, 30, 14, 19, 20, 24, 27, 9, 6, 12, 5, 1, 8, 7, 3, 13, 18, 26, 28, 23, 2, 11, 17, 22, 15], [7, 1, 18, 10, 2, 14, 21, 17, 3, 29, 27, 4, 6, 12, 20, 8, 22, 15, 16, 25, 11, 23, 9, 30, 26, 19, 5, 13, 24, 28], [26, 7, 18, 1, 21, 14, 10, 12, 17, 25, 29, 3, 22, 8, 16, 13, 15, 20, 4, 23, 27, 2, 19, 28, 11, 30, 24, 6, 5, 9], [6, 5, 30, 18, 14, 27, 7, 29, 13, 4, 26, 23, 3, 21, 8, 11, 25, 16, 22, 20, 10, 28, 2, 9, 15, 12, 1, 17, 24, 19], [4, 3, 7, 6, 15, 24, 27, 19, 20, 9, 1, 30, 23, 2, 5, 13, 21, 18, 28, 29, 25, 22, 17, 10, 26, 14, 16, 12, 8, 11], [16, 18, 28, 26, 12, 10, 21, 20, 3, 22, 5, 6, 13, 9, 25, 24, 23, 4, 19, 8, 14, 2, 11, 17, 7, 30, 27, 1, 29, 15], [15, 10, 14, 4, 11, 5, 3, 13, 24, 9, 2, 27, 6, 8, 12, 28, 18, 1, 7, 26, 25, 17, 29, 21, 30, 20, 22, 16, 19, 23]].

\vspace{1em}

Next, we show two different realizations of the push-up dynamics, each resulting from different pairs of manipulators at each time step. The resulting $\NE$ profiles correspond to different matchings under $\DA$. At each time step in the dynamics, we indicate the manipulating pair and the manipulated list of the accomplice man.

\vspace{0.5em}

\textbf{Dynamics 1:}

Accomplice man is  22 and the manipulator women is  1 and the changed preference list is  [25, 9, 22, 30, 19, 27, 28, 15, 14, 4, 8, 29, 13, 20, 26, 16, 1, 21, 24, 23, 2, 10, 17, 12, 3, 7, 5, 6, 11, 18]

Accomplice man is  1 and the manipulator women is  6 and the changed preference list is  [25, 9, 22, 30, 6, 27, 28, 15, 14, 4, 8, 29, 23, 20, 26, 16, 1, 21, 24, 19, 2, 10, 17, 12, 3, 7, 5, 13, 11, 18]

Accomplice man is  9 and the manipulator women is  8 and the changed preference list is [25, 9, 22, 30, 6, 27, 28, 15, 14, 4, 10, 29, 23, 20, 8, 16, 1, 21, 24, 19, 2, 26, 17, 12, 3, 7, 5, 13, 11, 18]

Accomplice man is  8 and the manipulator women is  21 and the changed preference list is  [25, 9, 22, 30, 6, 27, 28, 15, 14, 4, 10, 26, 23, 20, 8, 16, 1, 29, 24, 19, 2, 5, 17, 12, 21, 7, 3, 13, 11, 18]

\vspace{1em} 

Interestingly, woman 21 did not have an incentive to manipulate in the truthful profile and the possibility of obtaining a better partner only arose after the manipulations by the first three pairs.

\vspace{1em}

\textbf{Dynamics 2:} Here, woman 1 does not use man 22 as her accomplice, instead uses man 27.

Accomplice man is  27 and the manipulator women is  1 and the changed preference list is  [25, 9, 22, 30, 19, 27, 28, 15, 14, 4, 8, 29, 13, 20, 26, 16, 1, 21, 24, 23, 2, 10, 17, 12, 3, 7, 5, 6, 11, 18]

Accomplice man is  1 and the manipulator women is  6 and the changed preference list is  [25, 9, 22, 30, 6, 27, 28, 15, 14, 4, 8, 29, 23, 20, 26, 16, 1, 21, 24, 19, 2, 10, 17, 12, 3, 7, 5, 13, 11, 18]

Accomplice man is  1 and the manipulator women is  8 and the changed preference list is  [25, 9, 22, 30, 6, 27, 28, 15, 14, 4, 10, 29, 23, 20, 8, 16, 1, 21, 24, 19, 2, 26, 17, 12, 3, 7, 5, 13, 11, 18]

\vspace{1em}

In the $\DA$ matching for these two $\NE$, women 3, 5, 7, 21, 26, 29 are matched to different partners.
In Dynamic 1, the first accomplice man 22 eventually incurred regret as compared to his truthful partner even though he did not incur regret when he manipulated. Thus, in Dynamic 2 where man 22 does not manipulate first, after the manipulations by other men, he is not able to no-regret manipulate. But the manipulation by man 22 was crucial for woman 21 to receive a better partner in Dynamic 1. 

\section{Omitted Material from Section \ref{sec:transferable} }\label{sec:appendix_sec5}

\subsection{One-for-Many Manipulation Games}
\oneforall*
\begin{proof}
    Towards a contradiction, let $\succ_{\NE}$ be a strategy profile that is a Nash equilibrium in the accomplice manipulation game but not a Nash equilibrium in the one-for-all game. This means that there exists a man $m\in P_m$ who can manipulate and obtain a Pareto dominant outcome for everyone in $P_w$; let $w$ be a woman who strictly benefits from $m$'s misreport. This means that $(m,w)\in P$ is an accomplice pair at $\succ_\NE$, contradicting the fact that $\succ_\NE$ is a Nash equilibrium for the accomplice game. 
\end{proof}

\subsection{Woman Self-Manipulation Games}

Here, we provide the proof of Theorem \ref{cor:NE_existence_woman} and of the required lemmata. 
Given an instance $\langle M, W, \succ \rangle$, the players are the set of manipulating women $P_w \subseteq W$ and a strategy profile $\stratprrof$ is a list of strict total orderings over agents in $M$ submitted by all women in $P_w$. 
%
Given a strategy profile $\stratprof$, a \emph{self manipulation} by $w$ is a misreport $\stratprrof_w$ such that, for $\mu_\stratprof \coloneqq \DA(\succ_{-P_w}, \stratprof) $ and $\mu_{\stratprrof} \coloneqq \DA(\succ_{-P_w}, \stratprof_{-w},\stratprrof_w)$, we have $\mu_{\stratprrof} \succ_w \mu_{\stratprof}$. Then, a strategy profile $\succ_\NE$ is said to be a Nash equilibrium ($\NE$) if there does not exist a woman $w\in P_w$ such that $w$ can perform a self manipulation at the profile $\succ_\NE$.



%
Given any two strategy profiles $\stratprof, \stratprrof$, for every woman $w$ we can represent the preference list of $w$ in the strategy profile $\stratprrof$ based on the $\DA$ matching corresponding to the profile $\stratprof$ (i.e., $\mu_\stratprof$) as $\stratprrof_w = (\mathcal{L}(\stratprrof_w, \mu_\stratprof), \mu_{\stratprof}(w), \mathcal{R}(\stratprrof_w, \mu_\stratprof))$, where $\mu_\stratprof(w)$ is $w$'s partner under the matching $\mu_\stratprof$. Here, $\mathcal{L}(\stratprrof_w, \mu_\stratprof)$ is the set of men that $w$ reports as better than $\mu_{\stratprof}(m)$ in $\stratprrof_m$ and $ \mathcal{R}(\stratprrof_w, \mu_\stratprof)$ is the set of men who are ranked lower than $\mu_{\stratprof}(w)$ in $\stratprrof_w$.
The notion of a push-up operation $\stratprof_w^{X \uparrow}$ by lifting a set $X$ of men in the preference list of $w$ above her match in $\stratprof$ is also lifted from the case of men manipulation, and we say that it is inconspicuous if $|X|=1$.
%
%
%
Let us recall several facts regarding optimal manipulations by a self-manipulating woman. 

\begin{proposition}\label{prop:vaish}[Theorems 4 and 5, \cite{vaish2017manipulating}] 
    Given an instance $\langle M, W, \succ \rangle$, any optimal manipulation $\succ'$ by a woman can be performed in an inconspicuous manner. Furthermore, $\mathcal{S}_{\succ'} \subseteq \mathcal{S}_\succ$.
\end{proposition}

As in the case of accomplice manipulations, we construct a best-response dynamics which, when it converges, finds a Nash equilibrium of the woman self-manipulation game.

\paragraph{Inconspicuous dynamics.} For a stable matching instance $\langle M, W, \succ \rangle$ and a set of strategic women $P_w$, an inconspicuous dynamic starting at $\stratprof^{(1)} =\, \succ$ is a sequence of strategy profiles $\overline{\stratprof}= (\stratprof^{(1)}, \ldots)$, where each subsequent profile $\stratprof^{(t+1)}$ {results from an optimal inconspicuous self-manipulation by some strategic woman $w\in P_w$ in the matching instance $\langle M, W, \stratprof^{(t)}\rangle$, if possible.}  Otherwise, $\stratprof^{(t+1)} = \stratprof^{(t)}$.


In the next lemma, we show some properties of the inconspicuous dynamics; these help us establish that this sequence always converges and the converged point is an $\NE$.

\begin{lemma} \label{lem:woman_properties}
    Take a set of strategic women $P_w$ in the woman self-manipulation game, and an inconspicuous dynamic $\overline{\stratprof} = (\stratprof^{(1)}, \stratprof^{(2)}, \ldots)$. Then, when $\stratprof^{(t+1)} \neq \stratprof^{(t)}$,  the following set inclusion and set equalities hold for each woman $w$
 (1) $\mathcal{S}_{\stratprof^{(t+1)}} \subseteq \mathcal{S}_{\stratprof^{(t)}} \subseteq \ldots \subseteq \mathcal{S}_{\stratprof^{(1)}} = \mathcal{S}_\succ$, and
        (2) $\mathcal{R}(\stratprof^{(t+1)}_w, \mu_{\stratprof^{(t+1)}})= \mathcal{R}(\stratprof^{(t)}_w, \mu_{\stratprof^{(t+1)}} ) = \ldots = \mathcal{R}(\succ_w, \mu_{\stratprof^{(t+1)}})$.
   
\end{lemma}

With these properties, we show that the inconspicuous dynamics converge.

\womanNE*

\begin{proof}
Notice first that if $\stratprof^{(t+1)}\neq\stratprof^{(t)}$, then there is some woman $w\in P_w$ who can manipulate (optimally, inconspicuously) in the instance $\langle M, W, \stratprof^{(t)}\rangle$. By Lemma \ref{lem:woman_properties}, we have that $\mathcal{S}_{\stratprof^{(t+1)}} \subseteq \mathcal{S}_{\stratprof^{(t)}}$. By Proposition \ref{prop:accomplice_properties}(iii), 
\begin{align*}
    \mu_{\stratprof^{(t)}} \succ_M^{\stratprof^{(t)}} \mu_{\stratprof^{(t+1)}}.
\end{align*} 

Since in the inconspicuous dynamics of self manipulations by women, men do not change their preference profiles, we have that $\succ_M^{\stratprof^{(t)}} =\, \succ_M$ for every $t$ in the dynamic. Thus, $\mu_{\stratprof^{(t)}} \succ_M \mu_{\stratprof^{(t+1)}}$, and, by potential argument, the set of men can be worse-off with respect to $\succ$ only finitely many times. Thus, at some point $\stratprof^{(t+1)}=\stratprof^{(t)}$.

Notice that at every $t$ in which $\stratprof^{(t+1)} \neq \stratprof^{(t)}$, at least one man get strictly worse-off. He can get worse-off only $n$ many times, and there are only $n$ such men. Thus, the sequence must converge by $i \leq n^2$. 

When $\stratprof^{(t+1)}=\stratprof^{(t)}$, then $\stratprof^{(t)}$ is a fixed point and there is no strategic woman $w\in P_w$ who can perform an inconspicuous self manipulation in the instance $\langle M, W, \stratprof^{(t)}\rangle $. Thus, if the converged point $\stratprof^{(t)}$ is not an $\NE$, then some woman $w \in P_w$ can perform a self manipulation at $\stratprof^{(t)}$; let this manipulation be $\succ'_w$.

Consider the matching corresponding to $\succ'$ which is $\mu'=\DA(\succ'_w, {\stratprof^{(t+1)}}_{-w})$. Since it is a beneficial manipulation, $w$'s partner is such that $\mu'(w)\notin \mathcal{R}(\succ, \stratprof^{(t)}, w)$. But from Lemma 17 (ii), we know that $\mathcal{R}(\succ, \stratprof^{(t)}, w) = \mathcal{R}(\stratprof^{(t)}, \stratprof^{(t)}, w)$. Thus, $\mu'(w)\notin \mathcal{R}(\stratprof^{(t)}, \stratprof^{(t)}, w)$. However, this contradicts that $w$ can not perform a manipulation in the instance $\langle M, W, \stratprof^{(t)} \rangle$. 
\end{proof}

\section{Omitted Experimental Results}\label{sec:appendix_sec6}
In this section, we present the missing plots from \cref{sec:experiments}. 

\paragraph{Length of dynamics.} As mentioned, we study the impact of instance size and correlated preferences on the length of the dynamics.

\cref{fig:length_n} presents results on the impact of instance size on the length of the dynamics for both, accomplice and self manipulation games.
\begin{figure}[htp] 
    \centering 
        \includegraphics[width=.38\textwidth]{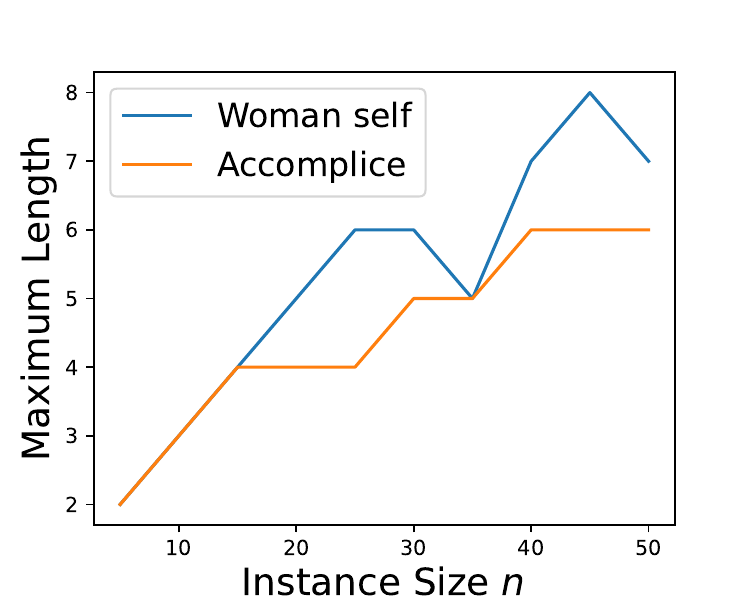}
        ~
         \includegraphics[width=.38\textwidth]{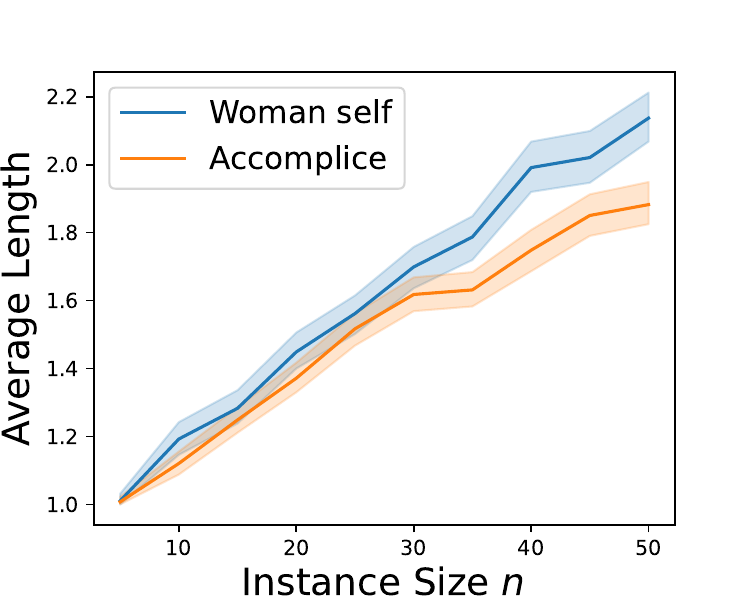} 

    \caption{Length of the constructed dynamics as a function of the instance size $n$. The plot on the left gives the maximum length across the samples for every considered instance size.  The plot on the right considers the average length across the samples as a function of $n$. 
    }
    \label{fig:length_n}
\end{figure}

Next, we present the plots for the effect of correlated preferences, measured through the dispersion parameters in Mallows' model, on the length of the push-up dynamics in the accomplice manipulation game (see \cref{fig:accomplice_phi}). 

\begin{figure}[h!] 
    \centering 
    \begin{subfigure}{0.4\textwidth} 
        \centering
        \includegraphics[width=\textwidth]{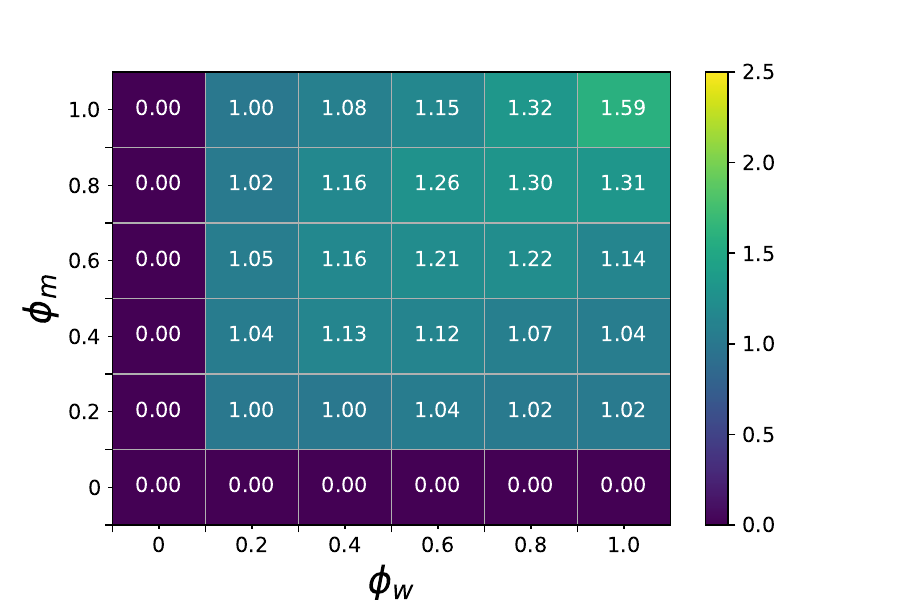} 
    \end{subfigure}
    \begin{subfigure}{0.4\textwidth}
        \centering
        \includegraphics[width=\textwidth]{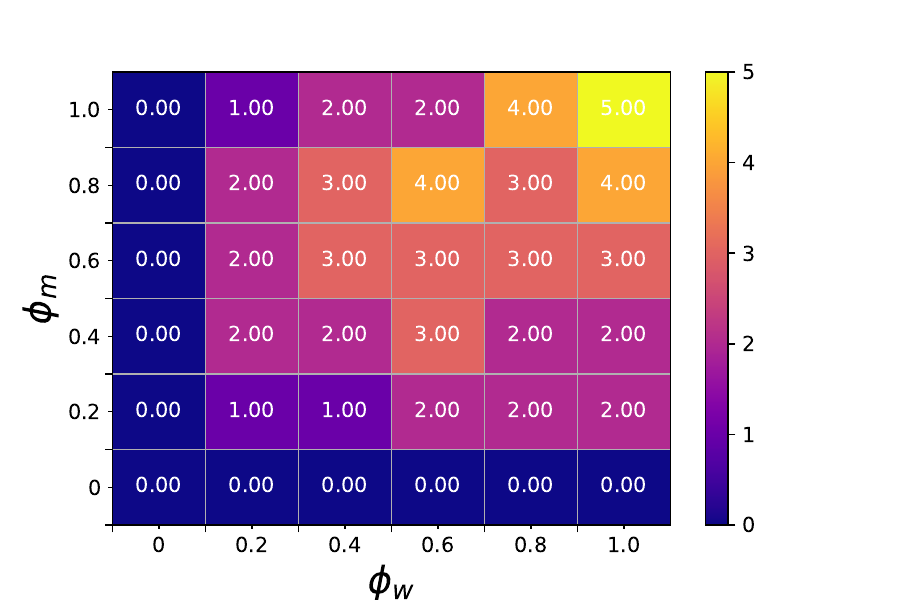} 
    \end{subfigure}
    
    \caption{Length of push-up dynamics in accomplice manipulation games as a function of the dispersion parameters $(\phi_w, \phi_m)$. The left plot is for the average length, whereas the one on the right is for the maximum length across the samples.}
    \label{fig:accomplice_phi}
\end{figure}

The equivalent plots for the woman self-manipulation game and its corresponding inconspicuous dynamics are presented next (see \cref{fig:length_phi}).

\begin{figure}[h!] 
    \centering 
    \includegraphics[width=.4\textwidth]{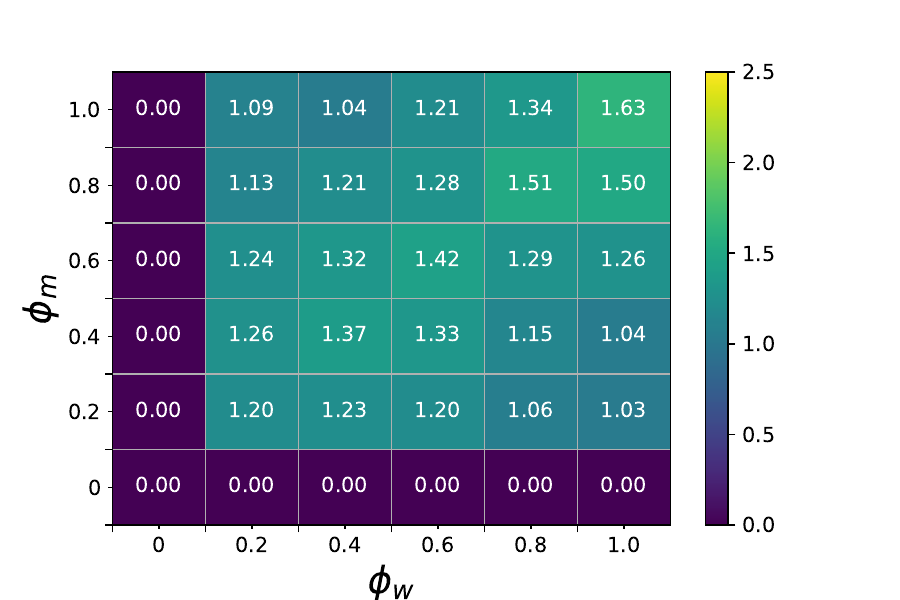} 
     ~            \includegraphics[width=.4\textwidth]{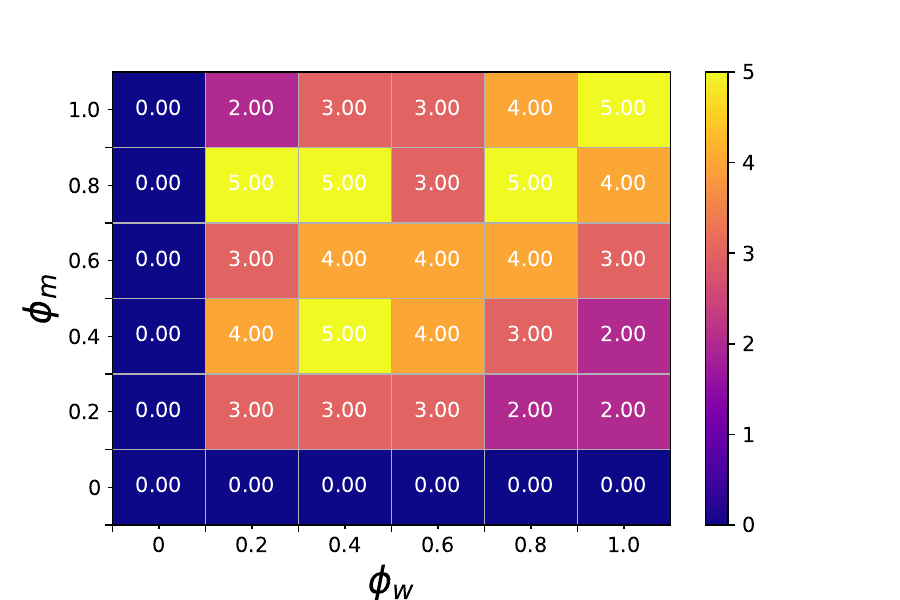} 

    \caption{Length of inconspicuous dynamic in woman self-manipulation games as a function of dispersion parameters $(\phi_w, \phi_m)$. 
    }
    \label{fig:length_phi}
\end{figure}

\paragraph{Welfare of agents.} 

First, we plot the welfare of the best-off woman and the worst-off man in the accomplice and woman self-manipulation games in \cref{fig:accomplice_welfare} and \cref{fig:women_welfare} respectively.
 \begin{figure}[h!] 
     \centering 
     \begin{subfigure}{0.38\textwidth}
         \centering
         \includegraphics[width=\textwidth]{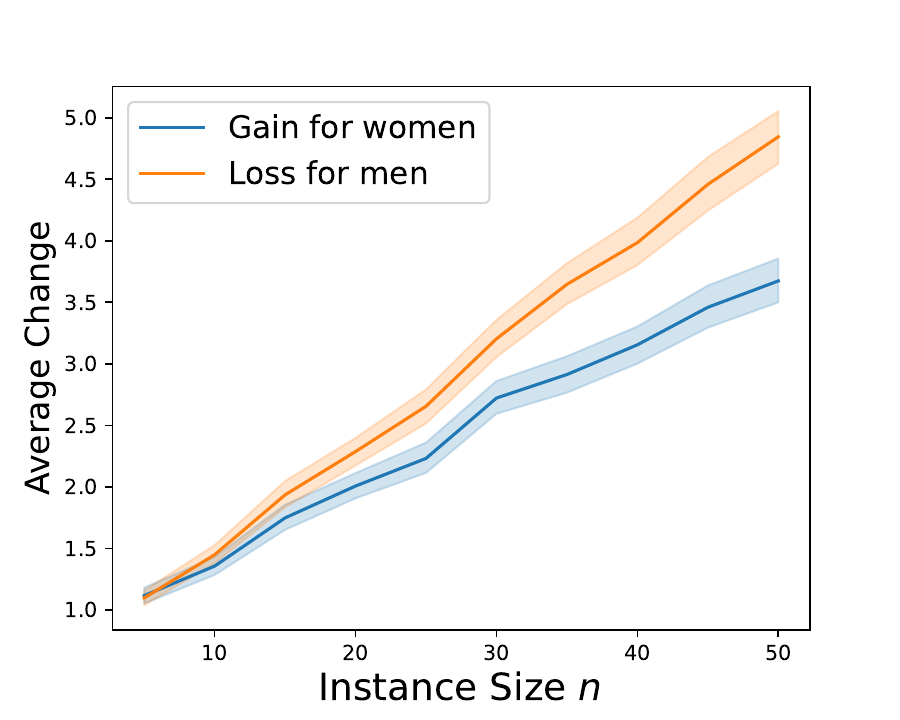} 
     \end{subfigure}
     \begin{subfigure}{0.38\textwidth} 
         \centering
         \includegraphics[width=\textwidth]{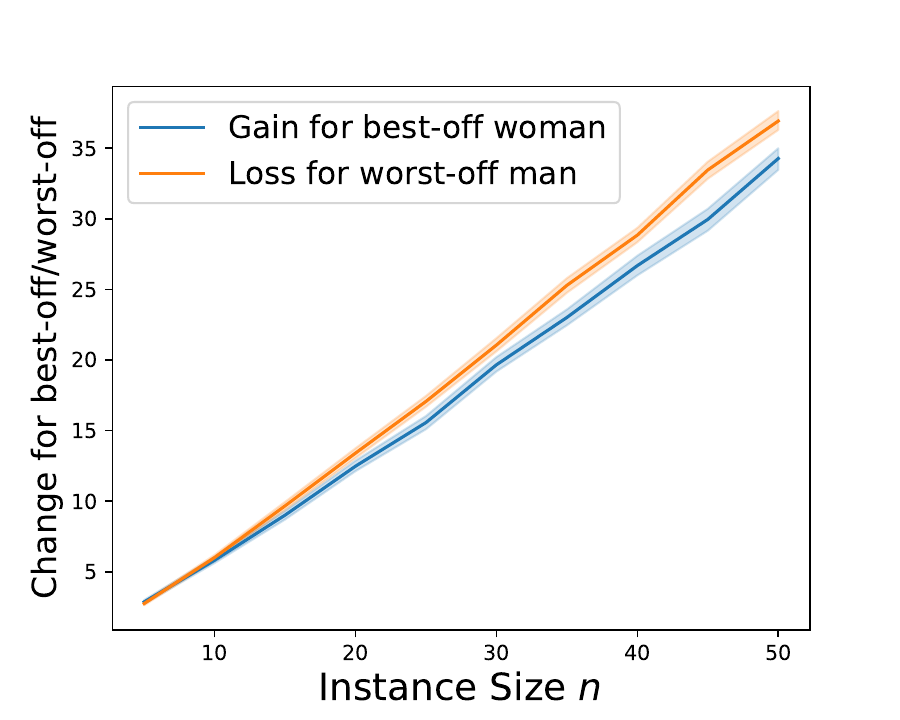} 
     \end{subfigure}
    
     \caption{Accomplice manipulation game: Women's gain and men's loss in $\NE$ profile as a function of the instance size. The left graph plots the gain for the average change, whereas the one on the right considers the change for the best-off and worst-off agent. }
     \label{fig:accomplice_welfare}
 \end{figure}

  \begin{figure}[h!] 
      \centering 
      \begin{subfigure}{0.38\textwidth} 
          \centering
          \includegraphics[width=\textwidth]{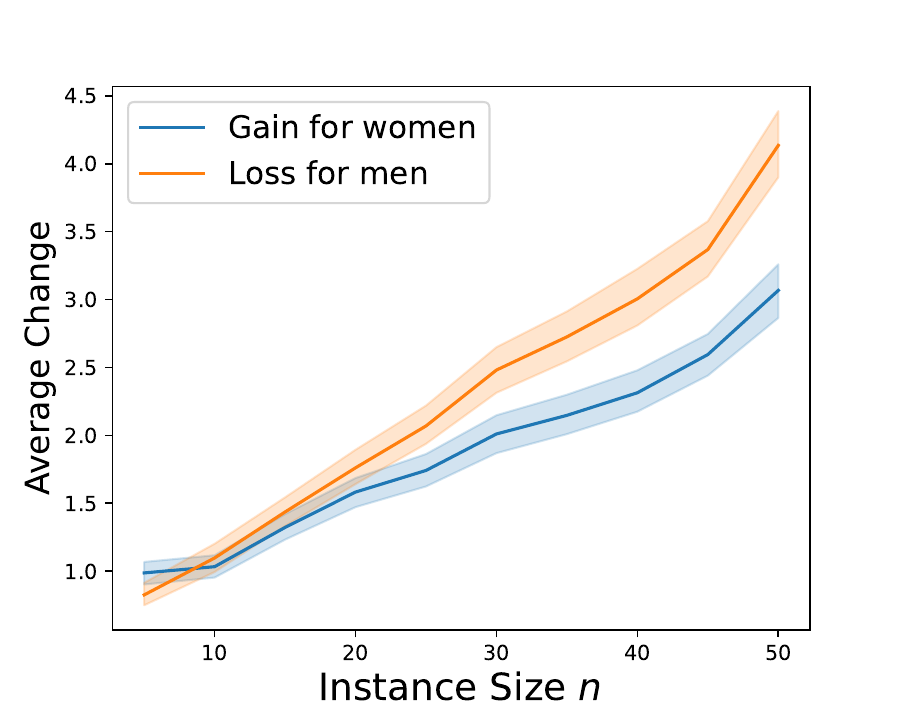} 
      \end{subfigure}
      \begin{subfigure}{0.38\textwidth}
          \centering
          \includegraphics[width=\textwidth]{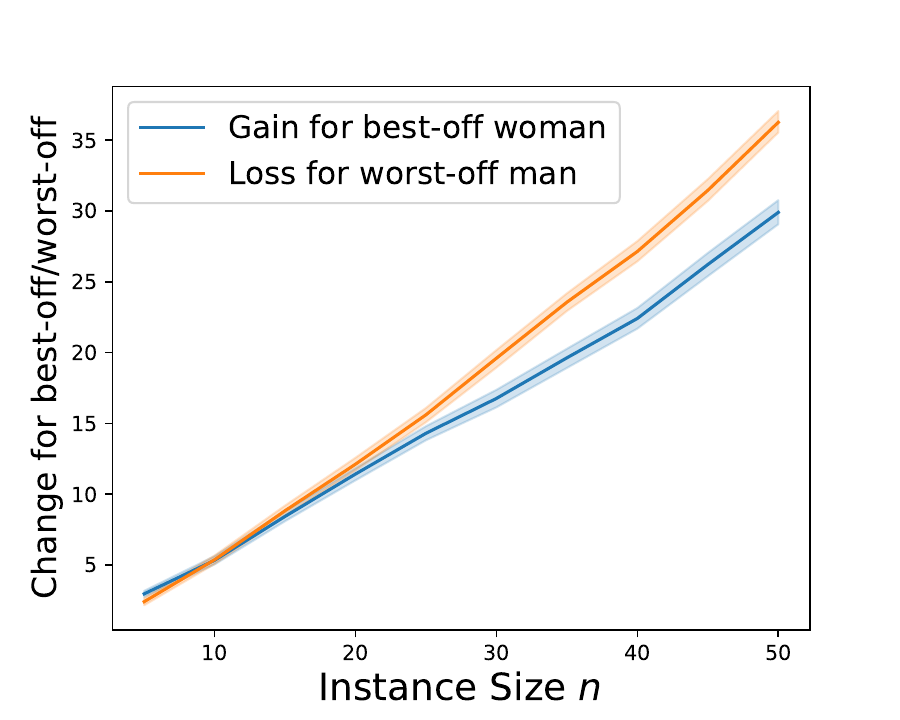} 
      \end{subfigure}
    
      \caption{Woman self-manipulation game: Women's gain and men's loss in $\NE$ profile as a function of the instance size. The left graph plots the gain for the average change, whereas the one on the right considers the change for the best-off and worst-off agent. }
      \label{fig:women_welfare}
  \end{figure}


Finally, the plots for the average gain by women in $\NE$ of the accomplice and self-manipulation game, as a function of correlated preferences, are presented in \cref{fig:avg_gain_mallows}.

 \begin{figure}[h!] 
      \centering 
      \begin{subfigure}{0.4\textwidth} 
          \centering
          \includegraphics[width=\textwidth]{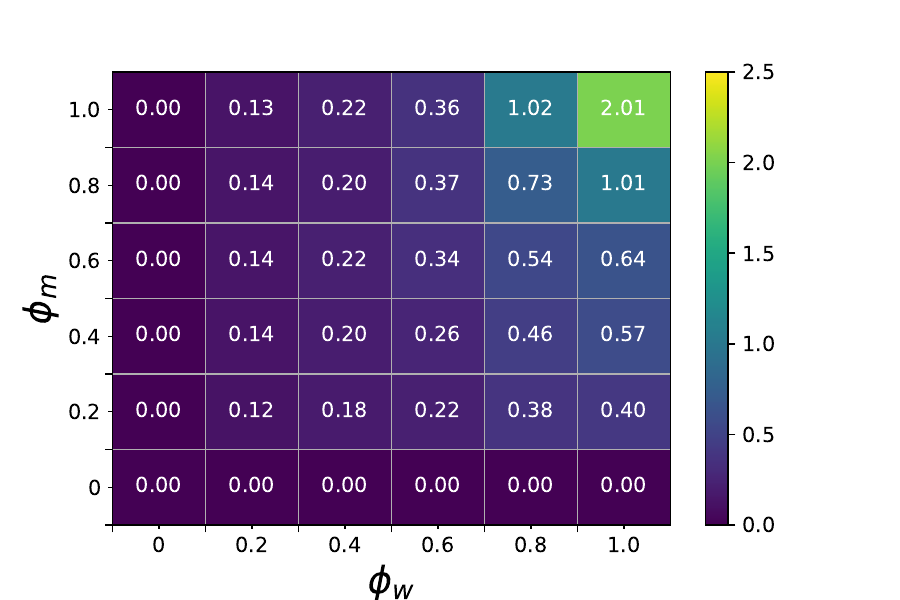} 
      \end{subfigure}
      \begin{subfigure}{0.4\textwidth}
          \centering
          \includegraphics[width=\textwidth]{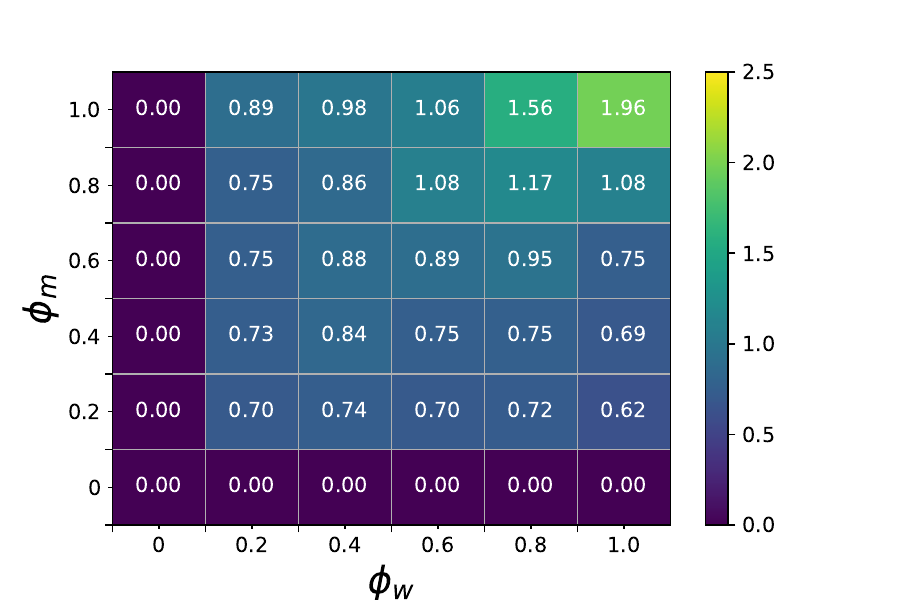} 
      \end{subfigure}
    
      \caption{Average gain by women in $\NE$ as a function of the dispersion parameters $(\phi_w, \phi_m)$. The plot on the left is for the accomplice manipulation game whereas the plot on the right is for the self-manipulation game.}
      \label{fig:avg_gain_mallows}
 \end{figure}

\end{document}